%% file: main.tex
\documentclass[a4paper,11pt]{article}

\usepackage{dsfont}
\usepackage{appendix}
\usepackage{extsizes}
\usepackage{lipsum}
\usepackage[colorlinks]{hyperref}
\usepackage[auth-lg]{authblk}
\usepackage{tcolorbox}
\usepackage{kvoptions}
\usepackage{xparse}
\usepackage{color}
\usepackage{etoolbox}
\usepackage[framemethod=TikZ, xcolor={RGB,table}]{mdframed}
\usepackage{mathtools}
\usepackage{fullpage}

\usepackage{fancyhdr}
\usepackage{colonequals}
\usepackage{yfonts}
\usepackage{graphicx}
\usepackage{amstext,amssymb,amsmath,amsthm}
\usepackage[nice]{nicefrac}

\usepackage{mathrsfs}
\usepackage{caption}
\usepackage{centernot}
\usepackage{tikz}
\usetikzlibrary{arrows}
\usepackage{microtype}
\usepackage{algorithm,algpseudocode}
\usepackage[flushmargin,norule]{footmisc}
\usepackage[margin=1.5cm]{geometry}
\usepackage{dsfont}
\usepackage{tcolorbox}

\usetikzlibrary{automata,positioning}

\usepackage{thmtools}
\usepackage{thm-restate}
\usepackage{cleveref}
 \usepackage{optidef}

\tcbuselibrary{theorems}

\hypersetup{linkcolor=blue,citecolor=blue,urlcolor=blue}

\tcbuselibrary{listings,breakable}
\allowdisplaybreaks

\DeclareMathOperator*{\Exp}{\ensuremath{{\mathbb{E}}}}
\DeclareMathOperator*{\Prob}{\ensuremath{\textnormal{Pr}}}
\renewcommand{\Pr}{\Prob}

\newcommand{\prob}[1]{\Pr\paren{#1}}
\newcommand{\card}[1]{|#1|}
\newcommand{\paren}[1]{\left ( #1 \right )}
\newcommand{\bracket}[1]{\left [ #1 \right ]}
\newcommand{\expect}[1]{\Exp\bracket{#1}}
\newcommand{\var}[1]{\textnormal{Var}\bracket{#1}}
\newcommand{\set}[1]{\ensuremath{\left\{ #1 \right\}}}
\newcommand{\poly}{\mbox{\rm poly}}
\newcommand{\eps}{\varepsilon}

\renewcommand{\algorithmicrequire}{\textbf{Input:}}
\renewcommand{\algorithmicensure}{\textbf{Output:}}

\theoremstyle{definition}
\newtheorem{definition2}{Definition}
\theoremstyle{definition}
\newtheorem{example2}[definition2]{Example}
\theoremstyle{definition}
\newtheorem{property}[definition2]{Property}
\theoremstyle{definition}
\newtheorem{theorem2}[definition2]{Theorem}
\theoremstyle{definition}
\newtheorem{lemma2}[definition2]{Lemma}
\theoremstyle{definition}
\newtheorem{corollary2}[definition2]{Corollary}
\theoremstyle{definition}
\newtheorem{observation2}[definition2]{Observation}
\theoremstyle{definition}
\newtheorem{claim2}[definition2]{Claim}
\theoremstyle{definition}
\newtheorem{note2}[definition2]{Note}
\theoremstyle{definition}
\newtheorem{remark2}[definition2]{Remark}
\theoremstyle{definition}
\newtheorem{research2}[definition2]{Research Direction}
\theoremstyle{definition}
\newtheorem{conjecture2}[definition2]{Conjecture}
\theoremstyle{definition}

\theoremstyle{definition}
\newtheorem{assumption}[definition2]{Assumption}
\theoremstyle{definition}
\usepackage{enumitem}
\setenumerate[0]{label=(\alph*)} 

\newenvironment{tbox}{\begin{tcolorbox}[
		enlarge top by=3pt,
		enlarge bottom by=3pt,
		boxsep=0pt,
		left=4pt,
		right=4pt,
		top=10pt,
		arc=0pt,
		boxrule=1pt,toprule=1pt,
		colback=blue!2,
		colframe=blue,
		]
	}
{\end{tcolorbox}}

\newenvironment{tbox2}{\begin{tcolorbox}[
		enlarge top by=3pt,
		enlarge bottom by=3pt,
		breakable,
		boxsep=0pt,
		left=4pt,
		right=4pt,
		top=10pt,
		arc=0pt,
		boxrule=1pt,toprule=1pt,
		colback=blue!2,
		colframe=blue,
		]
	}
	{\end{tcolorbox}}

\newtheorem{mdalg}{Algorithm}

\title{Decremental Matching in General Weighted Graphs}
\author{Aditi Dudeja\thanks{aditi.dudeja@plus.ac.at.}}
\affil{Department of Computer Science, University of Salzburg}
\date{}
\begin{document}
\maketitle

\abstract{In this paper, we consider the problem of maintaining a $(1-\eps)$-approximate \emph{maximum weight matching} in a dynamic graph $G$, while the adversary makes changes to the edges of the graph. In the fully dynamic setting, where both edge insertions and deletions are allowed, Gupta and Peng \cite{GP13} gave an algorithm for this problem with an update time of $\Tilde{O}_{\eps}(\sqrt{m})$. \\ \\ 
We study a natural relaxation of this problem, namely the decremental model, where the adversary is only allowed to delete edges. For the cardinality version of this problem in general (possibly, non-bipartite) graphs, \cite{AssadiBD22} gave a decremental algorithm with update time $O_{\eps}(\poly(\log n))$. However, beating $\Tilde{O}_{\eps}(\sqrt{m})$ update time remained an open problem for the \emph{weighted} version in \emph{general graphs}.\\ \\
In this paper, we bridge the gap between unweighted and weighted general graphs for the decremental setting. We give a $O_{\eps}(\poly(\log n))$ update time algorithm that maintains a $(1-\eps)$ approximate maximum weight matching under adversarial deletions.\footnote{Independently and concurrently, Chen, Sidford, and Tu \cite{ChenST23} obtained faster algorithms for decremental weighted matching in general graphs. Their work also improves on the state of the art for decremental unweighted matching in general graphs (that is, \cite{AssadiBD22}).}Like the decremental algorithm of \cite{AssadiBD22}, our algorithm is randomized, but works against an adaptive adversary. It also matches the time bound for the cardinality version upto dependencies on $\eps$ and a $\log R$ factor, where $R$ is the ratio between the maximum and minimum edge weight in $G$. }

\input{introduction}

\paragraph{Notation} Throughout the paper, we will use $G$ to refer to the current version of the graph, and let $V$ and $E$ be the vertex and edge sets of $G$ respectively. The graphs we deal with are weighted, and we use $R$ to denote the ratio between the max-weight and min-weight edge. Additionally use $\textsf{mwm}(G)$ to denote the weight of the maximum weight matching of $G$. During the course of the algorithm, we will maintain a fractional matching, which corresponds to a non-negative vector $\vec{x}\in [0,1]^m$ satisfying the following constraints : $\sum_{v\ni e}x(e)\leqslant 1$. For a set $S\subseteq E$, we let $x(S)=\sum_{e\in S} x(e)$. Given a capacity function $\kappa $ on the edges of the graph, we say that $\vec{x}$ obeys $\kappa$ if $x(e)\leqslant \kappa(e)$ for all $e\in E$. For a vector $\vec{x}$, we use $\text{supp}(\vec{x})$ to denote the set of edges that are in the support of $\vec{x}$. For a fractional matching $\vec{x}$, we say that it satisfies odd set constraints if for all odd-sized sets $B\subseteq V$, we have, $\sum_{e\in G[B]}x(e)\leqslant \frac{\card{B}-1}{2}$.  We use $\textsf{mwm}(G,\kappa)$ to denote the size of the maximum weight fractional matching obeying the odd set constraints and the capacity function $\kappa$. Additionally, we will use $\gamma_{\eps}=(\nicefrac{1}{\eps})^{O(\nicefrac{1}{\eps})}$. \\ \\
Our main result is a decremental algorithm for maintaining $(1-\eps)$-approximate maximum matching in weighted general graphs. In particular, we restate our main theorem.

\mainthm*

\subsection*{Some Tools from the Literature}

In order to prove our theorem, we will need the following constant weight assumption due to \cite{GP13}. This will enable us to only focus on solving the decremental matching problem on graphs with integral weights in $\set{1,2,\cdots,W}$ as opposed to arbitrary weight graphs. We refer the reader to the Appendix C in \cite{BernsteinDL21} for a proof of their reduction. 

\begin{lemma2}\label{lem:guptapeng}
Let $G$ be a weighted graph with arbitrary weights with $R$ being the ratio between maximum and minimum edge weights. Let $\eps\in (0,0.5)$. If there is an algorithm $\mathcal{A}$ that maintains an $\alpha$-approximate maximum weight matching in a graph whose weights are $\set{1,2,3,\cdots, W}$ with update time $T(n,m,\alpha, W)$, then there is an algorithm $\mathcal{A}'$ that maintains a $(1-\eps)\cdot \alpha$-approximate maximum weight matching in $G$ in update time $O(T(n,m,\alpha, \gamma_{\eps})\cdot \log R)$. 
\end{lemma2}

Additionally, we need the following algorithm, which essentially states that in order to solve the problem on graphs with arbitrary weighted maximum matchings, it is sufficient to solve the problem on graphs which have a matching with a large weight. We state the following reduction which was used by \cite{AKL19} in the stochastic setting for unweighted graphs. Here, we use it for the case of weighted graphs. We postpone the proof to the \Cref{sec:redc}.

\begin{restatable}{lemma2}{simpletomulti}\label{lem:simpletomulti}
Let $\delta\in (0,0.5)$ and $G$ be a weighted graph with weights in $\set{1,2,\cdots, W}$. Let $\textsf{mwm}(G)\geqslant \mu$ and at most $(1+\delta)\cdot \mu$.  There is an algorithm $\textsc{Vertex-Red}(G,\mu, \delta)$ that in time $O(\nicefrac{m\cdot W^3\cdot \log n}{\delta^4})$ returns $\lambda = \frac{100\cdot \log n\cdot W^3}{\delta^4}$ multigraphs $H_1,\cdots, H_{\lambda}$ that have the following properties, where the second property holds with probability at least $1-\nicefrac{1}{n^{2}}$. 
\begin{enumerate}
\item For all $i\in [\lambda]$, $\card{V(H_i)}=\frac{4\cdot (1+\delta)\cdot W\cdot \textsf{mwm}(G)}{\delta}$, $E(H_i)\subseteq E(G)$, and \label{item:parta}
\item Suppose $G$ is subject to deletions, which we simulate on each of the $H_i$'s. More concretely, if an edge $e$ is deleted from $G$, then its corresponding copy in $H_i$ is also deleted. Suppose at the end of the deletion process, $\textsf{mwm}(G)\geqslant \Tilde{\mu}$ for some $\Tilde{\mu}\geqslant (1-2\delta)\cdot \mu$, then $\textsf{mwm}(H_i)\geqslant (1-\delta)\cdot \Tilde{\mu}$ for some $i\in [\lambda]$. \label{item:partb}
\end{enumerate}
\end{restatable}
We also need the following result about weighted matchings due to \cite{DP14}.
\begin{lemma2}\label{lem:staticmatch}
There is an $O(\nicefrac{m}{\eps}\cdot \log \nicefrac{1}{\eps})$ time algorithm \textsc{Static-Weighted-Match} that takes as input a graph $G$, and returns a weighted integral matching such that $w(M)\geqslant \textsf{mwm}(G)\cdot (1-\eps)$. 
\end{lemma2}

Finally, since we are settling for an approximate matching, we can say the following:

\begin{observation2}\label{obs:ratiominmaxweight}
We can assume that the ratio between the minimum and the maximum edge weight is at least $\nicefrac{\eps^2}{n^2}$. This is because throwing out all the edges below this threshold will only reduce the weight of the matching by $o(\eps)$. 
\end{observation2}

\begin{assumption}\label{assumption:largematching}
In the rest of the paper we will assume that $\textsf{mwm}(G)=\Omega(\log n)$. This is because the graph is undergoing deletions, and soon as \textsf{mwm}(G) drops below this value, we can switch to a different decremental algorithm which we describe in \Cref{section:smallmatching}.
\end{assumption}

\subsection*{Main Result}

We first state the following lemma, and then assuming this show \Cref{thm:main}. Finally, we state the algorithm for that proves this lemma. 

\begin{lemma2}\label{lem:intermmain}
Let $G$ be a weighted graph with edge weights in $\set{1,2,\cdots, W}$ and let $\eps\in (0,1)$. Then, there is a decremental algorithm with total update time $O(m\cdot 2^{\nicefrac{W}{\eps^2}}\cdot \poly(\log n))$ that maintains an integral matching $M$ with $w(M)\geqslant (1-\eps)\cdot \textsf{mwm}(G)$, with high probability. The algorithm is randomized but works against an adaptive adversary. 
\end{lemma2}

Now, we use \Cref{lem:intermmain} to show \Cref{thm:main}.

\begin{proof}[Proof of \Cref{thm:main}]
By the above lemma, we know that for integer-weighted graphs, with weights in $\set{1,2,\cdots, W}$, there is a decremental $(1-\eps)$-approximate maximum weight matching algorithm with a total update time of $O(m\cdot 2^{\nicefrac{W}{\eps^2}}\cdot \poly(\log n))$. By applying \Cref{lem:guptapeng}, we know there is a decremental $(1-\eps)$-approximate maximum weight matching algorithm for arbitrary weight graphs with maximum edge weight $R$, with an update time of $O(m\cdot 2^{(\nicefrac{1}{\eps})^{O(\nicefrac{1}{\eps})}}\cdot\poly(\log n)\cdot\log R)$. 
\end{proof}

Our next step is the following theorem, which we will use to prove \Cref{lem:intermmain}.

\begin{restatable}{theorem2}{decmatching}\label{thm:main2}
We are given a weighted multigraph $G=(V,E)$ with edge weights in $\set{1,2,\cdots, W}$. Let $\card{V}=n$ and $\card{E}=m$. Suppose $\textsf{mwm}(G)\geqslant \nicefrac{\varepsilon\cdot n}{100\cdot W}$ and $\mu\geqslant \textsf{mwm}(G)\cdot (1-\varepsilon)$. Suppose $G$ is subject to adversarial deletions. There is an algorithm \textsc{Dec-Matching}($G,\mu, \varepsilon$) that processes deletions in $O(m\cdot 2^{\nicefrac{W}{\eps^2}}\cdot \poly(\log n))$ time, and has the following guarantees:
\begin{enumerate}
\item When the algorithm terminates, $\textsf{mwm}(G)< (1-2\varepsilon)\cdot \mu$. Upon termination, the algorithm outputs ``\textsc{no}". \label{item:main2a}
\item Until the algorithm terminates, it maintains an integral matching $M$ with $w(M)\geqslant (1-20\varepsilon)\cdot \mu$. \label{item:main2b}
\end{enumerate}
\end{restatable}
The decremental algorithm mentioned in \Cref{lem:intermmain} uses \textsc{Dec-Matching}() stated in \Cref{thm:main2} as a subroutine. We first show the proof of \Cref{lem:intermmain} assuming \Cref{thm:main2}, and then we show the proof of \Cref{thm:main2}.

\paragraph{Description of Algorithm for \Cref{lem:intermmain}} The algorithm takes as input a graph $G$, and a parameter $\eps>0$. It first estimates the weight of the maximum matching of $G$ by running \textsc{Static-Weighted-Match}($G,\eps$) given in \Cref{lem:staticmatch}. Let this estimate be $\mu$. The algorithm then instantiates a procedure to create multiple parallel instances of 
the algorithm \textsc{Dec-Matching}(). As mentioned in \Cref{thm:main2}, \textsc{Dec-Matching}() can maintain a $(1-\eps)$-approximate maximum weight matching of a multigraph, provided the weight of the matching in the multigraph is $\Omega(\nicefrac{n}{W})$. So, we use the reduction mentioned in \Cref{lem:simpletomulti} to create $\lambda$ multigraphs $H_1,\cdots, H_{\lambda}$ which have $O(\nicefrac{W\cdot \mu}{\eps})$ vertices as opposed to $n$ vertices, and we run parallel instances of $\textsc{Dec-Matching}(H_i,\mu\cdot(1-\eps),\eps)$ for all $i\in [\lambda]$. This procedure also instantiates a pointer called \textsc{cur}, which points to the least indexed $H_i$ that contains a matching of weight at least $(1-\eps)\cdot \mu $. The algorithm will output the matching indexed by \textsc{cur} (denoted $M_{\textsc{cur}}$). The algorithm then moves onto a procedure for deletion of edges. If an edge $e\in G$ is deleted, then it is deleted from each of the $H_i$'s labelled \textsc{active}. If at this point, the weight of the maximum matching of $H_{\textsc{cur}}$ drops by a significant amount, then we increment the value of $\textsc{cur}$. Finally, if all $\textsf{mwm}(H_i)$ drops by a significant amount for all $i\in [\lambda]$, then we can conclude that $\textsf{mwm}(G)$ has also dropped by a significant amount (this is due to contrapositive of \Cref{lem:simpletomulti}\ref{item:partb}). 

\begin{algorithm}
	\algorithmicrequire{ Graph $G$ and a parameter $\eps>0$}\\
	\algorithmicensure{ A matching $M$ with $w(M)\geqslant (1-8\eps)\cdot \textsf{mwm}(G)$}
	\caption{}
	\begin{algorithmic}[1]
		\State $M\leftarrow\textsc{Static-Weighted-Match}(G,\eps)$. \label{item:line1}
		\State $\mu\leftarrow w(M)$
        \State $\lambda \leftarrow \frac{W^3\cdot \log n}{\eps^4\cdot (1-\eps)}$
		\Procedure{Instantiating \textsc{Dec-Matching}()}{}
			\State $i\leftarrow 1$.
			\State \textsc{cur}$\leftarrow \lambda +1$.
			\State Let $H_1,H_2,\cdots ,H_{\lambda}$ be the multigraphs returned by $\textsc{Vertex-Red}(G,\mu,\eps)$. 
			\State Label $H_1,\cdots, H_{\lambda}$ as \textsc{active}.
			\For{$i\leqslant \lambda$}
			\If{$w(\textsc{Static-Weighted-Match}(H_i,\eps))<(1-\eps)\cdot \mu$}
			\State Label $H_i$ as \textsc{inactive}. 
			\EndIf
			\EndFor
			\For{$H_i$ labelled \textsc{active}}
			\State Initialize $\textsc{Dec-Matching}(H_i,\mu\cdot (1-\eps),\eps)$ (denoted $\mathcal{A}_i$).
			\State Let $M_i$ be the matching maintained by $\mathcal{A}_i$. 
			\EndFor
			\State Let \textsc{cur} be the least index $i$ such that $H_i$ is \textsc{active}. 
			\If{\textsc{cur} $>\lambda$} 
			\State \textsc{terminate}.
			\EndIf
		\EndProcedure
		\Procedure{Handling deletion of edge $e$}{}
		\For{$H_i$ labelled \textsc{active}}
		\State Feed the deletion of $e$ to $\mathcal{A}_i$. \label{alg1:dec-matching}
		\If{$\mathcal{A}_i$ returns \textsc{no}}
		\State Label $H_i$ as \textsc{inactive}.
		\EndIf
		\EndFor
		\If{$H_{\textsc{cur}}$ labelled \textsc{active}}\label{line:curactive} \Comment{Label may be changed in \Cref{alg1:dec-matching}}.
		\State Continue to output $M_{\textsc{cur}}$. 
		\Else
		\If{$\textsc{cur}>\lambda$}
			\label{line:restart}	\State Return to \Cref{item:line1}. \Comment{Start over with new estimate for $\mu$}.
		\Else
		\State $\textsc{cur}\leftarrow \textsc{cur}+1$
		\State Goto \Cref{line:curactive}.
		\EndIf
		\EndIf
		\EndProcedure	
	\end{algorithmic}
 \caption{Algorithm for \Cref{lem:intermmain}}
	\label{alg:proofthm1}
\end{algorithm}

\begin{proof}[Proof of \Cref{lem:intermmain}]
We want to show how \Cref{alg:proofthm1} implies \Cref{lem:intermmain}. We first start with the proof of correctness. We want to claim that \Cref{alg:proofthm1} indeed outputs a matching of size at least $(1-10\eps)\cdot \textsf{mwm}(G)$. To see this, observe that if we run \Cref{alg:proofthm1} with input $G,\eps$ then, we want to claim that it returns a matching $M$ with $w(M)\geqslant (1-10\eps)\cdot \mu$, where $\mu\geqslant (1-\eps)\cdot \textsf{mwm}(G)$. Consider the multigraphs $H_1,\cdots,H_{\lambda}$ output by \textsc{Vertex-Red}($G,\eps$). By \Cref{lem:simpletomulti}, we know that these satisfy the requirements of \Cref{thm:main2}. Thus, every run of \textsc{Dec-Matching}($H_i,\mu,\eps$), until it returns \textsc{no}, it returns a matching of weight at least $(1-\eps)\cdot \mu\geqslant (1-10\eps)\cdot \textsf{mwm}(G)$.  This takes care of the correctness, and now we want to upper bound the runtime of the algorithm.

\begin{claim2}
Each time \Cref{alg:proofthm1} executes Line \ref{line:restart}, $\textsf{mwm}(G)$ has dropped by a factor of $(1-2\eps)$.
\end{claim2}
\begin{proof}
To see this, first observe that each time the algorithm returns to \Cref{item:line1}, it must be the case that \textsc{Dec-Matching}($H_i,\mu\cdot (1-\eps),\eps$) has returned \textsc{no} for all $i\in [\lambda]$. Thus, by \Cref{thm:main2}, we can conclude that $\textsf{mwm}(H_i)<(1-2\eps)\cdot (1-\eps)\mu$ for all $i\in [\lambda]$. This in turn implies by contrapositive of \Cref{lem:simpletomulti}\ref{item:partb} that $\textsf{mwm}(G)<(1-2\eps)\cdot \mu$ as well. Thus, $\textsf{mwm}(G)$ has dropped from at least $\mu$ to at most $(1-2\eps)\cdot \mu$. This concludes the proof of this claim. 
\end{proof}

We upper bound the runtime of the algorithm as follows: We instantiate $\lambda=O\paren{\frac{W^3\cdot \log n}{\eps^4\cdot(1-\eps)}}$ instances of \textsc{Dec-Matching}(). We start by bounding the runtime of the procedure that instantiates these $\lambda$ instances. This has its runtime dominated by \textsc{Vertex-Red}($G,\mu,\eps$), which takes $O(m)$ time. The second source of runtime is running the algorithm \textsc{Static-Weighted-Match}($H_i,\eps$), for $i\in [\lambda]$. \\ \\
Next, we upper bound the runtime of the procedure that handles deletions. Note that if an edge $e$ is deleted then, it is deleted from each of the $\lambda$ copies. Thus, the total time to process deletions over $m$ deletions is at most $O(m\cdot \lambda)$. Additionally, for maintaining each of the $\mathcal{A}_i$ for $i\in [\lambda]$, we take $O(m\cdot 2^{\nicefrac{W}{\eps^2}}\cdot \poly(\log n)\cdot \lambda)$. Finally, by the above claim, we can conclude that \Cref{line:restart} is executed $O(\log n)$ times, so the total runtime of \Cref{alg:proofthm1} is $O(m\cdot 2^{\nicefrac{W}{\eps^2}}\cdot \poly(\log n))$.
\end{proof}

\section{Algorithm for \Cref{thm:main2}}
Before we give the algorithm for \Cref{thm:main2}, we give some definitions. Recall, that is sufficient to design \textsc{Dec-Matching}() for the case of graph that have integer weights in $\set{1,2,\cdots,W}$ by \Cref{lem:guptapeng}. 

\begin{definition2}
For a pair of vertices $u,v$, we define $D_i(u,v)$ to be the edges $e$ between $u$ and $v$ that have weight $i$. Additionally, we also have $D(u,v)=\cup_{i=1}^W D_{i}(u,v)$. Since we assume integral weights, these sets are well-defined. If $e$ is an edge between $u,v$, then $D_i(e)\coloneqq D_i(u,v)$ and $D(e)\coloneqq D(u,v)$.
\end{definition2}

\begin{definition2}
Let $G$ be a weighted multigraph, with $n$ vertices and $m$ edges. Let $\kappa$ be the capacity function on the edges of the graph, and let $\vec{x}$ be a fractional matching. We define $\vec{x}^C$ to be a vector, with support size $\min\set{{n\choose 2}, m}$, where for a pair of vertices $u,v$, $x^C_{i}(u,v)=\sum_{e\in D_i(u,v)}x(e)$. That is, $\vec{x}^C$ is obtained by ``collapsing" all edges of the same weight between a pair of vertices together. Similarly, suppose $\vec{y}$ is a vector with $\card{\text{supp}(\vec{y})}\leqslant {n\choose 2}\cdot W$, where, $y_i(u,v)$ is the entry corresponding to the edge of weight $i$ between $u$ and $v$. Then, we define $\vec{y}^D$ to be an $m$ length vector such that for every $e\in E$ between $u,v$ with $w(e)=i$, $y^D(e)\coloneqq \frac{y_i(u,v)\cdot \kappa(e)}{\kappa(D_{i}(u,v))}$. Intuitively, $
\vec{y}^D$ is a multigraph obtained by distributing the $y_i(u,v)$ among all edges of the same weight.
\end{definition2}

\begin{definition2}
Throughout the algorithm, we use $\alpha_{\eps}\coloneqq 2^{\nicefrac{W}{\eps^3}}\cdot \log n$ and similarly, we have, $\rho_{\eps}\coloneqq 2^{\nicefrac{W}{\eps^2}}\cdot \log n$.
\end{definition2}

We now state the main lemma, which will be used to show \Cref{thm:main2}.

\begin{restatable}{lemma2}{MorE}\label{lem:MorE}
Let $G$ be a multigraph with edge weights in $\set{1,2,\cdots, W}$. Then, there is an algorithm \textsc{WeightedM-or-E*}() that takes as input $G,\kappa, \eps\in (0,0.5)$, and a parameter $\mu\geqslant \textsf{mwm}(G)\cdot (1-\eps)$, and in time $O(\nicefrac{m\cdot W\cdot \log n}{\eps})$ returns one of the following.
\begin{enumerate}
\item A fractional matching $\vec{x}$ such that $\sum_{e\in E}w(e)\cdot x(e)\geqslant (1-2\eps)\cdot \textsf{mwm}(G)$, with the following properties. \label{item:MorEa}
\begin{enumerate}[label=(\roman*)]
\item Let $e\in D_{i}(u,v)$ such that $e\in \text{supp}(\vec{x})$, with $\kappa(D_i(u,v))\geqslant \nicefrac{1}{\alpha_{\eps}^2}$, $x(D_i(u,v))=1$, and we have, $x(e)=\frac{\kappa(e)}{\kappa(D_i(u,v))}$. Moreover, since $\vec{x}$ is a fractional matching, we have $x(D_j(u,v))=0$ for all $j\neq i$. \label{item:MorEai}
\item Consider any $e\in D_{i}(u,v)$ such that $e\in \text{supp}(\vec{x})$, with $\kappa(D_i(u,v))\leqslant \nicefrac{1}{\alpha_{\eps}^2}$, then $x(D_i(u,v))\leqslant \kappa(D_i(u,v))\cdot \alpha_{\eps}$, and $x(e)\leqslant \kappa(e)\cdot \alpha_{\eps}$. Moreover, for any $e\in D_j(u,v)$ with $\kappa(D_j(u,v))>\nicefrac{1}{\alpha_{\eps}^2}$, we have $x(e)=0$.\label{item:MorEaii}
\end{enumerate}
\item A set of edges $E^*$ such that $\sum_{e\in E^*}w(e)\cdot \kappa(e)=O(\textsf{mwm}(G)\cdot \log n)$, and for any integral matching $M$ with $w(M)\geqslant (1-2\eps)\cdot \textsf{mwm}(G)$, we have $w(M\cap E^*)\geqslant \eps\cdot \textsf{mwm}(G)$. Additionally, for all $e\in E^*$, we have $\kappa(e)<1$. \label{item:MorEb}
\end{enumerate}
\end{restatable} 

Consider the situation in which the algorithm $\textsc{WeightedM-or-E*}()$ returns a fractional matching $\vec{x}$. For any pair of vertices $u,v$ consider $\sum_{i=1}^W x_{i}^C(u,v)$, then this is either $1$ or at most $\eps$ by \Cref{lem:MorE}\ref{item:MorEa}. Thus, it satisfies odd set constraints for all odd sets of size at most $\nicefrac{1}{\eps}$, by a folklore lemma (\Cref{obs:litteflow}), we can then argue that $\vec{x}$ contains an integral matching of weight at least $(1+\eps)^{-1}\cdot \sum_{u\neq v}\sum_{i=1}^W x^{C}_{i}(u,v)\cdot i$ in its support. 

\begin{property}
We will use \textsc{WeightedM-or-E*}() as a subroutine in \textsc{Dec-Matching}(), and we will get a matching $\vec{x}$ with the following properties. We mention these properties since they will help us visualize the fractional matching better. 
\begin{enumerate}
\item Each time \textsc{WeightedM-or-E*}() returns $E^*$, we will increase capacity along $E^*$ by multiplying existing capacities by the same factor.
\item Consider $u,v\in V$, and let $e,e'$ be two edges between $u$ and $v$ such that $w(e)=w(e')$, then $\kappa(e)=\kappa(e')$ at all times during the run of the algorithm. 
\end{enumerate}
\end{property}
 The next property follows directly as a consequence of \Cref{lem:MorE}\ref{item:MorEa}.
 
 \begin{property}\label{prop:excessflow}
Let $\vec{x}$ be a matching output by \Cref{lem:MorE}, then $x(e)\leqslant \kappa(e)\cdot \alpha_{\eps}^2$ for all $e\in E$. This is evident from \Cref{lem:MorE}\ref{item:MorEa}.
 \end{property}
 
 \begin{definition2}
 Let $G$ be a multigraph, and let $\vec{x}$ be a fractional matching of $G$. Then, we split supp($\vec{x}$) into two parts: $\vec{x}^i$ and $\vec{x}^f$, where, $\vec{x}=\vec{x}^i+\vec{x}^f$, and $\text{supp}(\vec{x}^i)=\set{e\mid x(D_j(e))>\nicefrac{1}{\alpha_{\eps}^2}\text{ and }w(e)=j}$ and $\text{supp}(\vec{x}^f)=\set{e\mid x(D_j(e))\leqslant \nicefrac{1}{\alpha_{\eps}^2}\text{ and }w(e)=j}$. When $\vec{x}$ is the output of \textsc{WeightedM-or-E*}(), then these correspond to the integral and fractional parts of $\vec{x}$. 
 \end{definition2}
 
 \begin{property}\label{prop:vertexdisjoint}
 Let $G$ be any multigraph, and let $\vec{x}$ be a fractional matching of $G$ that is returned by \textsc{WeightedM-or-E*}(). Then, for any pair of vertices, $u,v$, we have the following properties, which are a consequence of \Cref{lem:MorE}\ref{item:MorEa}\ref{item:MorEai} and \Cref{lem:MorE}\ref{item:MorEa}\ref{item:MorEaii}:
 \begin{enumerate}
 \item Either $D(u,v)\cap \text{supp}(\vec{x}^i)\neq \emptyset$ or $D(u,v)\cap \text{supp}(\vec{x}^f)\neq \emptyset$, but not both. 
 \item For any $u,v,j$ such that $D_j(u,v)\cap \text{supp}(\vec{x})\neq \emptyset$, we have $D_j(u,v)\subseteq \text{supp}(\vec{x})$.
 \end{enumerate}
 Thus, the supports of $\vec{x}^i$ and $\vec{x}^f$ are vertex disjoint.
  \end{property}
  \begin{property}\label{prop:integpart}
  Let $\vec{x}$ be a fractional matching returned by \textsc{WeightedM-or-E*}(), consider $\vec{z}=\vec{x}^i$. Then, $\vec{z}^C$ is an integral matching. This is implied by \Cref{lem:MorE}\ref{item:MorEa}. 
  \end{property}
   The next ingredient is a method to round weighted fractional matchings. This is implicit in \cite{Wajc2020} (for unweighted, bipartite graphs), however in the \Cref{sec:rounding}, we extend show this for the case of general, weighted graphs for the specific case of the fractional matching used by our algorithm. 
   
   \begin{restatable}{lemma2}{lemuno}\label{lem:sparsification}(\cite{Wajc2020})
   Suppose $G$ is a integer-weighted multigraph, let $\eps\in (0,\nicefrac{1}{2})$, and let $\vec{x}$ be a fractional matching such that $x(D(e))\leqslant \eps^6$ for all $e\in E$. Then, there is a dynamic algorithm \textsc{Sparsification}($\vec{x},\eps$), that has the following properties.
   \begin{enumerate}
   \item \label{lem:sparseb} The algorithm maintains a subgraph $H\subseteq \text{supp}(\vec{x})$, such that $\card{E(H)}=O_{\eps}(\textsf{mwm}(\text{supp}(\vec{x}))\log n)$ and with high probability, $\textsf{mwm}(H)\geqslant (1-\eps)\cdot \sum_{e\in E}w(e)\cdot x(e)$.
   \item \label{lem:sparsea} The algorithm handles the following updates to $\vec{x}$: the adversary can either remove an edge from supp($\vec{x}$) or for any edge $e$, the adversary can reduce $x(e)$ to some new value $x(e)\geqslant 0$. 
   \item The algorithm handles the above-mentioned updates in $\Tilde{O}_{\eps}(1)$ worst case time. 
   \end{enumerate}
   \end{restatable}

We now show how \textsc{WeightedM-or-E*}() and \textsc{Sparsification}() give us an algorithm for \Cref{thm:main2}. This algorithm is called \textsc{Dec-Matching}(). 

\paragraph{Description of \textsc{Dec-Matching}()}
The algorithm takes as input a graph $H$, and instantiates its capacities. It then gets an estimate on $\textsf{mwm}(H)$, denoted $\mu$. The algorithm then runs the subroutine \textsc{WeightedM-or-E*}($H,\mu, \eps,\kappa$), and this subroutine either outputs a fractional matching $\vec{x}$ with $\sum_{e\in E}w(e)\cdot x(e)\geqslant (1-\eps)\cdot \textsf{mwm}(G)$, or if the matching is not heavy enough, a set of edges $E^*$ (see \Cref{lem:MorE}). In the latter case, the algorithm proceeds by increasing the capacities along $E^*$ (see \Cref{line:capchange}). On the other hand, if the fractional matching is large enough, then by \Cref{lem:MorE}\ref{item:MorEa}, it can be broken up into two parts where $x^C((u,v))=1$ or $x^C((u,v))\leqslant \eps^6$ (see \Cref{line:yz}). The algorithm \textsc{Sparsification}() extracts the fractional half of the matching (\Cref{line:sparsification}). Note that this part satisfies the premise of \Cref{lem:sparsification}, and therefore, the output is a sparse graph $S$, where $\textsf{mwm}(S)\geqslant (1-\eps)\sum_{e\in E} w(e)\cdot x(e)$. Thus, we then run \textsc{Static-Weighted-Match}($S,\eps$) and combine this with the integral part of $\vec{x}$ to output $M$ (see \Cref{line:rounding}). \\ \\
Subsequently, the algorithm then processes deletions. It maintains two counters, namely, \textsc{CounterM} and \textsc{CounterX}. As edges are deleted from the graph, at some point if \textsc{CounterX} increases to $\eps\cdot \mu$, then we know that the adversary has reduced the weight of the fractional matching by a significant amount. At this point: either $\textsf{mwm}(G)$ has dropped by a significant amount, or it hasn't and we must compute a fresh robust fractional matching, and therefore we begin a new phase. On the other hand, if \textsc{CounterX} has not increased, but \textsc{CounterM} has increased to $\eps\cdot \mu$, then we know that $\vec{x}$ contains a high weight integral matching its support, and we recompute it.
\begin{algorithm}[H]
	\caption{\textsc{Dec-Matching}($H,\mu,\eps$)}
	\label{alg:decmatching}
	\begin{algorithmic}[1]
		\State \label{line:initcap}For all $e\in E(H)$, $\kappa(e)\leftarrow \nicefrac{1}{\alpha_{\eps}^{\lceil \log_{\alpha_{\eps}}n \rceil}}$
		\Procedure{Starting a new phase}{} 
		\State $\mu'\leftarrow$ \textsc{Static-Weighted-Match}($H,\eps$)\label{line:recomp}\Comment{$\mu'\geqslant (1-\eps)\cdot \textsf{mwm}(H)$}
		\If{$\mu'\leqslant (1-3\eps)\cdot \mu$} 
		\State Return \textsc{no}, and terminate.\label{line:funcinfty}
		\Else 
		\While{\textsc{WeightedM-or-E*}($H,\kappa,\eps,\mu'$) returns $E^*$}
		\State \label{line:capchange} For all $e\in E^*$, $\kappa(e)\leftarrow \kappa(e)\cdot \alpha_{\eps}$
		\EndWhile
		\EndIf
		\EndProcedure
		\State Let $\vec{x}$ be the matching returned by \textsc{WeightedM-or-E*}(). \Comment{This is a matching in a multigraph.}
		\State $\vec{y}\leftarrow \vec{x}^f$, $\vec{z}\leftarrow \vec{x}^i$ \label{line:recomp2}\Comment{We will update $\vec{y},\vec{z}$ as edges are deleted.}\label{line:yz}
		\State $S\leftarrow \textsc{Sparsification}(\vec{y}^C,\eps)$ \Comment{Recall that \textsc{Sparsification}() is dynamic} \label{line:sparsification}
		\State \label{line:rounding}$M\leftarrow \textsc{Static-Weighted-Match}(S,\eps)\cup \text{supp}(\vec{z}^C)$. \Comment{The matching $M$ that is output.}
		\State $\textsc{CounterM}\leftarrow 0$ \Comment{Counter for deletions to the integral matching.}
		\State $\textsc{CounterX}\leftarrow 0$ \Comment{Counter for deletions to the fractional matching.}
		\Procedure{For deletion of edge $e$ from $H$}{}
		\If{$e\in \text{supp}(\vec{x})$} 
		\State Delete $e$ from $\text{supp}(\vec{x})$
		\State Update $\vec{z}^C$ accordingly. \Comment{See \Cref{line:yz}.}
		\State \label{line:sparsification-update}Update $\vec{y}^C$ accordingly; \textsc{Sparsification}($\vec{y}^C,\eps$) from Line \ref{line:sparsification} then updates $S$.
		\State \textsc{CounterX} $\leftarrow$ \textsc{CounterX} $+\ w(e)\cdot x(e)$.
		\EndIf
		\If{\textsc{CounterX} $>\eps\cdot \mu$} \label{line:counterx}
		\State Goto \Cref{line:recomp}  \Comment{End current phase and start new one.}
		\EndIf
		\If{$e\in M$}
		\State Delete $e$ from $M$
		\State \textsc{CounterM} $\leftarrow$ \textsc{CounterM} $+\ w(e)$ 
		\EndIf
		\If{$\textsc{CounterM}>\eps\cdot \mu$} \Comment{The matching $M$ is out of date}
		\State $\textsc{CounterM}\leftarrow 0$
		\State \label{line:rebuild}$M\leftarrow \textsc{Static-Weighted-Match}(S,\eps)\cup \text{supp}(\vec{z}^C)$ \Comment{Recompute outputted matching $M$}
		\EndIf
		\EndProcedure
	\end{algorithmic}
\end{algorithm}

We now restate and prove \Cref{thm:main2}.

\decmatching*

\begin{proof}[Proof of \Cref{thm:main2}]
We first argue that \Cref{alg:decmatching} satisfies the properties of \Cref{thm:main2}. First note that when the \Cref{alg:decmatching} executes \Cref{line:funcinfty}, note that $\mu'\leqslant (1-3\eps)\cdot \mu$. Since $\mu'\geqslant \mu(H)$, this implies that $\mu(H)\leqslant (1-2\eps)\cdot \mu$. This already shows \ref{item:main2a}. We now want to show \ref{item:main2b}. In order to show this, first note \Cref{lem:MorE}\ref{item:MorEa} states that $\sum_{e\in E}w(e)\cdot x(e)\geqslant (1-\varepsilon)\cdot \textsf{mwm}(H)$. Note that the output matching $M$ is a combination of \text{supp}($\vec{z}^C$) and the integral matching computed on the output of \textsc{Sparsification}($\vec{y}^C, \eps$). Recall that \text{supp}($\vec{z}^C$) is a matching by \Cref{prop:integpart}. Also, by \Cref{prop:vertexdisjoint}, these matchings are on disjoint set of vertices. Now we want to show that $w(M)\geqslant (1-\varepsilon)\cdot\textsf{mwm}(G)$. \\ \\
In order to see this, first note that $\vec{y}^C$ satisfies the premise of \Cref{lem:sparsification}. Consequently, $\textsf{mwm}(S)\geqslant (1-\eps)\cdot \sum_{e\in E}w(e)\cdot y^C(e)$. Thus, we have the following line of reasoning:
\begin{align*}
w(M)&\geqslant \sum_{e\in \textsf{supp}(\vec{z}^C)}w(e)+w(\textsc{Static-Weighted-Match}(S,\eps))\\
&\geqslant \sum_{e\in \textsf{supp}(\vec{z}^C)}w(e)+(1-\eps)\cdot \textsf{mwm}(S)\\
&\geqslant \sum_{e\in \textsf{supp}(\vec{z}^C)} w(e)+(1-\eps)^2\cdot \sum_{e\in E}w(e)\cdot y^C(e)\\
&\geqslant (1-2\eps) \sum_{e\in E} w(e)\cdot x(e)\\
&\geqslant (1-3\eps)\cdot \textsf{mwm}(H)
\end{align*}
We use the following lemma to bound the runtime of \textsc{Dec-Matching}(). We will give a formal proof of this lemma in \Cref{sec:MorE}.
\begin{lemma2}\label{lem:callstoMorE}
The subroutine \textsc{WeightedM-or-E*}() and Line \ref{line:recomp} are called at most $O(\alpha_{\eps}^3\cdot \log^2n)$ times in \Cref{alg:decmatching} during the course of $m$ deletions. 
\end{lemma2}
\paragraph{Runtime} Finally, we want to analyze the runtime of this algorithm. First, from \Cref{lem:callstoMorE} we can conclude that \Cref{line:recomp} and \textsc{WeightedM-or-E*}() is called at most $O(\alpha_{\eps}^3\cdot \log^2n)$ many times. Moreover, from \Cref{lem:MorE}, we can conclude the total runtime of these routines is $O(m\cdot W\cdot\log n)$, thus the total runtime over all edge deletions is at most $O(m\cdot W\cdot\log^3 n\cdot \alpha_{\eps}^3)$.\\ \\
Next, we upper bound the number of phases. From \Cref{lem:callstoMorE}, we can conclude that the number of phases is upper bounded by $O(\alpha_{\eps}^3\cdot \log n)$. 
In every phase, we initialize the procedure \textsc{Sparsification}(). This takes $O(m\cdot \poly(\log n, \eps^{-1}))$ time to initialize. Thus, the total time over all phases, is $O(m\cdot \poly(\log n, \eps^{-1})\cdot \alpha_{\eps}^3)$. Additionally, \textsc{Sparsification}() is dynamic, and it maintains the sparsifier in $O(\poly(\log n, \nicefrac{1}{\eps}))$ update time. Finally, we also note that \Cref{line:rebuild} runs a static weighted matching algorithm on $S$ each time the total weight of $M$ drops by a $(1-\varepsilon)$ factor. This means, the graph has been subjected to at least $\frac{\eps\cdot \textsf{mwm}(G)}{W}$ deletions. Next, observe by \Cref{lem:sparsification}, $|E(S)|=O(\textsf{mwm}(H))$. Thus, from \Cref{lem:staticmatch}, we can conclude that \Cref{line:rebuild} takes amortized $O(W)$ time over all updates. 
\end{proof}

The subsequent section is dedicated to the proof of \Cref{lem:callstoMorE}.

\section{Proof of \Cref{lem:callstoMorE}}\label{sec:MorE}

In order to prove the unweighted version of \Cref{lem:callstoMorE}, \cite{AssadiBD22} and \cite{BPT20} use the framework of congestion balancing. Here, we show that this framework extends to the case of general weighted graphs as well.  

\begin{definition2}
Let $G$ be a multigraph, and let $\kappa$ be the capacity function on the edges of the graph (if the adversary deletes an edge $e$, then let $\kappa(e)$ denote its capacity right before it is deleted). Let $\mu$ be the input to \textsc{Dec-Matching}(), when it is run on $G,\eps$. Let $\mathcal{M}=\set{M\mid w(M)\geqslant (1-\eps)\cdot \textsf{mwm}(G)}$. Define the cost of an edge $e$, $c(e)=\log (n\cdot \kappa(e))$. For any integral matching $M$, define $\sum_{e\in M}c(e)\cdot w(e)=\sum_{e\in M}w(e)\cdot \log (n\cdot \kappa(e))$. We also define $\Pi(G,\kappa)=\min_{M\in \mathcal{M}}\sum_{e\in M} c(e)\cdot w(e)$. This corresponds to the min-cost matching in $\mathcal{M}$. If $\mathcal{M}=\emptyset$, then $\Pi(G,\kappa)=\infty$. 
\end{definition2}

\begin{observation2}
Initially, $\kappa(e)=\nicefrac{1}{\alpha_{\eps}^{\lceil 2\cdot\log_{\alpha_{\eps}}n \rceil}}$, so $\Pi(G,\kappa)=0$. Thereafter, since capacities can only be increased, and edges can only be deleted, $\Pi(G,\kappa)$ can only increase.
\end{observation2}

\begin{observation2}
When $\Pi(G,\kappa)=\infty$, then Line \ref{line:funcinfty} of \Cref{alg:decmatching} causes the algorithm to terminate. 
\end{observation2}

\begin{lemma2}\label{lem:upperbdpotential}
In \textsc{Dec-Matching}, we say that we have started a new phase, when $\eps\cdot\mu$ weight of the fractional matching has been deleted, that is, $\textsc{CounterX}$ has increased to $\eps\cdot \mu$ (see \Cref{line:counterx}). Suppose we are in a phase where the algorithm does not terminate and let $\kappa$ be the capacity function right before we process deletions, then, $\Pi(G,\kappa)=O(\mu\cdot \log n)$.
\end{lemma2}

\begin{proof}
First observe that for any edge $e\in E$, $\kappa(e)\leqslant 1$, so $c(e)\leqslant \log n$. This is implied by the fact that $E^*$ only increases capacity along edges with $\kappa(e)<1$ (see \Cref{lem:MorE}\ref{item:MorEb}), and due to the fact that initially,  $\kappa(e)=\nicefrac{1}{\alpha_{\eps}^{\lceil 2\cdot \log_{\alpha_{\eps}}n \rceil}}$, and each time we increase $\kappa(e)$, we multiply it by a factor of $\alpha_{\eps}$. Moreover, we know that for any matching $M$, $w(M)\leqslant (1+\eps)\cdot \mu$, since $\mu\geqslant (1-\eps)\cdot\textsf{mwm}(G)$. Thus, $\sum_{e\in M} w(e)\cdot c(e)=O(\mu \log n)$. This implies that $\Pi(G,\kappa)=O(\mu\log n)$. 
\end{proof}

\begin{definition2}
Let $E_0$ be the initial set of edges, then we let capacitated weight of $E_0$ to be defined as $w(\kappa(E_0))=\sum_{e\in E_0}w(e)\cdot \kappa(e)$. If $e$ is deleted by the adversary, then $\kappa(e)$ is the capacity of $e$ right before deletion.
\end{definition2}

\begin{lemma2}\label{lem:potupperbd}
Suppose a call to \textsc{WeightM-or-E*}($G,\mu,\kappa, \eps$) returns the set $E^*$ instead of a matching. Let $\kappa'$ denote the new edge capacities after increasing capacities along $E^*$. Then,
\begin{enumerate}
\item $w(\kappa'(E_0))\leqslant w(\kappa(E_0))+\alpha_{\eps}\cdot \mu\cdot \log n$, and \label{item:potupperbda}
\item $\Pi(G,\kappa')\geqslant \Pi(G,\kappa)+\nicefrac{\mu}{\eps}$.
\end{enumerate}
\end{lemma2}
\begin{proof}
We begin by recalling \Cref{lem:MorE}\ref{item:MorEb}, which states that $w(\kappa(E^*))=O(\mu\log n)$. Consequently, we have,
\begin{align*}
w(\kappa'(E_0))&=w(\kappa(E_0\setminus E^*))+w(\kappa'(E^*))\\
		    &\leqslant w(\kappa(E_0))+O(\mu\cdot \alpha_{\eps}\cdot \log n)
\end{align*}
The last inequality follows from \Cref{lem:MorE}\ref{item:MorEb}, and due to the fact that we multiply capacities along $E^*$ by $\alpha_{\eps}$. To prove the second part of the claim, we have, again use the fact that every $M\in \mathcal{M}$ has the property that $w(M\cap E^*)\geqslant \eps\cdot \textsf{mwm}(G)$. Consequently, we have for every matching $M\in \mathcal{M}$, $c(w(M))$ increases by $\sum_{e\in M\cap E^*}w(e)\cdot \log \alpha_{\eps}\geqslant \nicefrac{\textsf{mwm}(G)}{\eps}$. Thus, the second part of the claim follows. 
\end{proof}

\begin{observation2}\label{obs:calltype1}
The total number of calls to \textsc{WeightedM-or-E*}() (till \textsc{Dec-Matching}() terminates) that return $E^*$ are upper bounded by $O(\eps\cdot \log n)$. This is because each such call increases $\Pi(G,\kappa)$ by $\nicefrac{\mu}{\eps}$, and by \Cref{lem:upperbdpotential}, we know that $\Pi(G,\kappa)$ is upper bounded by $\mu\log n$. Additionally, we can conclude that $w(\kappa(E_0))\leqslant \alpha_{\eps}\cdot \mu\cdot \eps\cdot \log ^2n$, and this is due to \Cref{lem:potupperbd}\ref{item:potupperbda}, and the fact that the number of calls to \textsc{WeightedM-or-E*}() that return $E^*$ are at most $O(\eps\cdot \log n)$ many. 
\end{observation2}

\begin{lemma2}\label{lem:calltype2}
The total number of phases is at most $O(\alpha_{\eps}^3\log ^2n)$. Therefore, the total number of times, a call to $\textsc{WeightedM-or-E*}$() returns a fractional matching is at most $O(\alpha_{\eps}^3\log ^2n)$.
\end{lemma2}
\begin{proof}
Note that we begin a new phase when \textsc{CounterX} grows from $0$ to $\eps\cdot\mu$. That is, the weight of the current fractional matching $\vec{x}$ has dropped by $\eps\cdot\mu$. Let $E_{\textsf{del}}$ be the set of edges that have been deleted from the multigraph. Let $\Phi_{\textsf{del}}=\sum_{e\in E_{\textsf{del}}} \kappa(e)\cdot w(e)$. This upper bounded by $w(\kappa(E_0))\leqslant \alpha_{\eps}\cdot \mu\cdot \eps\cdot \log^2 n$. Note that each time we begin a new phase, $\Phi_{\textsf{del}}$ increases by at least $\nicefrac{\eps\mu}{\alpha_{\eps}^2}$ (one can conclude this since \Cref{prop:excessflow} holds). Thus, the total number of phases is at most $O(\alpha_{\eps}^3\cdot \log ^2n)$.
\end{proof}

\begin{proof}[Proof of \Cref{lem:callstoMorE}]
There are two types of calls to the algorithm \textsc{WeightedM-or-E*}(): ones that return $E^*$, and the others which return a good fractional matching. From \Cref{obs:calltype1} and \Cref{lem:calltype2}, we can conclude that the total number of calls to \textsc{WeightedM-or-E*}() is at most $O(\alpha_{\eps}^3\log^2 n)$.
\end{proof}

The rest of the paper is dedicated to giving an algorithm for \textsc{WeightedM-or-E*}(), and this culminates in a proof for \Cref{lem:MorE}, which is our main technical contribution.

\section{Building Blocks for \textsc{WeightedM-or-E*}()}

In this section, we prove the structural lemmas needed for each of the phases, and in order to do so, we need some probabilistic tools, for example Chernoff Bound:

\begin{lemma2}[\cite{DP09}]\label{lem:chernoff}
Let $X_1,X_2,\cdots,X_k$ be $k$ negatively correlated random variables, and let $X=\sum_{i=1}^k X_i$. Suppose $\mu=\expect{X}$, and $\mu_{\text{min}}\leqslant \mu \leqslant \mu_{\text{max}}$. Then, we have,
\begin{align*}
\prob{X\geqslant (1+\delta)\cdot \mu_{\text{max}}}\leqslant \paren{\frac{e^{\delta}}{\paren{1+\delta}^\delta}}^{\mu_{\text{max}}}.
\end{align*}
\end{lemma2}

Additionally, we will also need Bernstein's inequality.

\begin{lemma2}[\cite{MMR94}]\label{lem:concentrate}
	Let $X$ be the sum of negatively associated random variables $X_1,\cdots, X_k$ with $X_i\in [0,M]$ for each $i\in [k]$. Then, for $\sigma^2=\sum_{i=1}^k \var{X_i}$ and all $a>0$, 
	\begin{align*}
	\prob{X>\expect{X}+a}\leqslant \exp\paren{\frac{-a^2}{2\paren{\sigma^2+\nicefrac{a\cdot M}{3}}}}
\end{align*}
\end{lemma2}

\subsection{Phase 1 of \textsc{WeightedM-or-E*}()}

Recall that we used \textsf{mwm}($G,\kappa$) to denote the weight of the maximum weight fractional matching of $G$ obeying $\kappa$ as well as the odd set constraints. As in the congestion balancing step, we want to estimate $\textsf{mwm}(G,\kappa)$. However, existing techniques only work for bipartite or general unweighted graphs. In this section, we extend this to the case of general weighted graphs. In particular, we extend the structural lemma to show that if in a weighted graph $G$, we sample every edge with probability $p(e)=\min\set{1,\kappa(e)\cdot \rho_{\eps}}$ to create a graph $G_s$, then $\textsf{mwm}(G_s)\geqslant \textsf{mwm}(G,\kappa)-\eps\cdot \textsf{mwm}(G)$. Thus, now we can estimate $\textsf{mwm}(G_s)$ by using existing results. 

\begin{lemma2}\label{lem:phase1}
Let $G$ be an integer weighted multigraph with weights in $\set{1,2,\cdots, W}$,  and with $\textsf{mwm}(G)\geqslant \max\set{\nicefrac{\eps\cdot n}{16\cdot W},\nicefrac{\log n}{\eps^4}}$, where $\eps\in (0,\nicefrac{1}{2})$. Let $\kappa$ be a capacity function on the edges of the graph, and let $G_s$ be obtained by sampling every edge $e$ independently with probability $p(e)=\min\set{\kappa(e)\cdot \rho_{\eps},1}$. Let $\textsf{mwm}(G,\kappa)$ be the weight of the maximum weight fractional matching of $G$ that obeys the odd set constraints and the capacity constraints. Suppose $\textsf{mwm}(G,\kappa)\geqslant (1-\eps)\cdot \textsf{mwm}(G)$, then, with high probability, $\textsf{mwm}(G_s)\geqslant \textsf{mwm}(G,\kappa)-\eps\cdot \textsf{mwm}(G)$.
\end{lemma2}


\begin{proof}[Proof of \Cref{lem:phase1}]
Let $\vec{x}$ denote the fractional matching that realizes \textsf{mwm}($G,\kappa$). In order to prove the statement, we will construct a vector $\vec{z}$ in the support of $G_s$. This vector $\vec{z}$ will have the following properties: it will satisfy the fractional matching constraints and small blossom constraints with high probability. Additionally, it will also have the property that $\sum_{e\in E)}w(e)\cdot z(e)\geqslant (1-\eps)\cdot \sum_{e\in E}w(e)\cdot x(e)\geqslant \textsf{mwm}(G,\kappa)-\eps\cdot \textsf{mwm}(G)$ with high probability as well.\\ \\
Since with high probability $G_s$, contains a fractional matching $\vec{z}$ satisfying the above properties, there is an integral matching $M$ in support of $G_s$ with $w(M)\geqslant (1-\eps)\cdot \sum_{e\in E}w(e)\cdot z(e)\geqslant \textsf{mwm}(G,\kappa)-\eps\cdot \textsf{mwm}(G)$ (see \Cref{obs:litteflow}). \\ \\
Let $X_e$ denote a random variable that takes value $1$ if $e\in H_s$ and value $0$ otherwise. We define $\vec{z}$ as follows:
\begin{align*}
z(e)=X_e\cdot \frac{x(e)}{\min\set{1,\kappa(e)\cdot \rho_{\eps}}}\cdot (1-\eps)
\end{align*}
Note that $z(e)$ are independent random variables. The following is true.
\begin{align*}
\expect{\sum_{e\in E}z(e)\cdot w(e)}=(1-\eps)\cdot \sum_{e\in E} w(e)\cdot x(e)
\end{align*}
Since this is a sum of independent random variables, using Chernoff bound (see \Cref{lem:chernoff}), we have,
\begin{align*}
\prob{\sum_{e\in E}z(e)\cdot w(e)\leqslant \sum_{e\in E} w(e)\cdot x(e)-2\cdot\eps\cdot \textsf{mwm}(G)}=\exp\paren{\nicefrac{-\eps^2\textsf{mwm}(G)}{2}}
\end{align*}
By assumption, $\textsf{mwm}(G)\geqslant \nicefrac{100\cdot \log n}{\eps^4}$, we have that the above claim holds with probability at least $1-O(\nicefrac{1}{n^{\nicefrac{1}{\eps}}})$. We will now show that it obeys fractional matching constraints with high probability as well. Consider an edge $e$ with $\kappa(e)>\nicefrac{1}{\rho_{\eps}}$, we know that $e\in G_s$. Consequently, we have $\var{z(e)}=0$ for such an edge, since $z(e)=x(e)$ always. For any edge with $\kappa(e)<\nicefrac{1}{\rho_{e}}$, we have,
\begin{align*}
\var{z(e)}&\leqslant (1-\eps)^2\cdot  \paren{\frac{x(e)}{\kappa(e)\cdot \rho_{\eps}}}^2\cdot p(e)= (1-\eps)^2\frac{x(e)^2}{\kappa(e)\cdot \rho_{\eps}}\leqslant (1-2\eps)\cdot \frac{x(e)}{\rho_{\eps}}
\end{align*}
Additionally, since $z(e)$ are independent random variables, we have $\var{\sum_{v\ni e}z(e)}=\sum_{v\ni e} \var{z(e)}\leqslant (1-2\eps)\cdot \rho_{\eps}^{-1}$. Additionally, we know that $\expect{\sum_{e\ni v} z(e)}\leqslant (1-\eps)$. Thus, we want to compute the probability of the event that $\sum_{e\ni v} z(e)\geqslant 1$. Note that, in order to do this, it is sufficient to consider the edges $e$ for which $\kappa(e)<\rho_{\eps}^{-1}$, since for edges other than this, $z(e)=x(e)$. Thus, we have, by \Cref{lem:concentrate}, with $a=\eps$, $M=\rho_{\eps}^{-1}$,
\begin{align*}
\prob{\sum_{e\ni v}z(e)>1}&\leqslant \exp\paren{\frac{-\eps^2}{\frac{1}{\rho_{\eps}}+\frac{\eps}{\rho_{\eps}} }}=O\paren{\frac{1}{n^{\nicefrac{1}{\eps}}}}
\end{align*}
\end{proof}
Taking a union bound over all vertices in the graph, we have our claim. Next, we want to compute the probability that $\vec{z}$ also satisfies small odd set constraints. To see this, consider any odd set $B$ such that $\card{B}\leqslant \nicefrac{1}{\eps}$. We know that by definition of $\vec{x}$, we have,
\begin{align*}
\sum_{e\in G[B]} x(e)\leqslant \frac{\card{B}-1}{2} \text{ and,}\\
\expect{\sum_{e\in G[B]}z(e)}\leqslant (1-\eps)\cdot \frac{\card{B}-1}{2}
\end{align*}
Similarly, we can bound variance as well,
\begin{align*}
\var{\sum_{e\in G[B]} z(e)}\leqslant (1-\eps)\cdot \frac{\card{B}-1}{2}\cdot \frac{1}{\rho_{\eps}}
\end{align*}
Consider the following subclaim.
\begin{observation2}
    Let $B$ be any odd set such that $3\leqslant \card{B}\leqslant \nicefrac{1}{\eps}$. Then,
    \begin{align*}
    \frac{\card{B}-1}{2}=(1-\eps)\cdot \frac{\card{B}-1}{2}+\eps\cdot\frac{\card{B}-1}{2}\geqslant (1-\eps)\cdot \frac{\card{B}-1}{2} +\eps
    \end{align*}
\end{observation2}
Let $\mathcal{E}_{B}$ be the event that $\sum_{e\in G[B]}z(e)\geqslant \frac{\card{B}-1}{2}$,
\begin{align*}
\prob{\mathcal{E}_{B}}&\leqslant \prob{\sum_{e\in G[B]}z(e)\geqslant (1-\eps)\cdot \frac{\card{B}-1}{2}+\eps}\\
&\leqslant \exp\paren{-\frac{\eps^2}{\frac{\card{B}-1}{2}\cdot \frac{1}{\rho_{\eps}}+\frac{\eps}{\rho_{\eps}}}}\\
				  &=O\paren{\frac{1}{n^{{\nicefrac{1}{\eps}}}}}	
\end{align*}
The second equality follows from the fact that we are considering small blossoms, that is, $\card{B}\leqslant \nicefrac{1}{\eps}$. Taking a union bound over all small blossoms, we have our claim. 

\subsection{Phase 2 of \textsc{WeightedM-or-E*}()}

The algorithm \textsc{WeightedM-or-E*}() proceeds to Phase 2 only if in Phase 1, if $\textsf{mwm}(G,\kappa)\geqslant (1-\eps)\cdot \textsf{mwm}(G)$ (\Cref{lem:phase1}). In Phase 2, we now want to construct a good fractional matching. 

\begin{definition2}
Let $G$ be a multigraph, we define $E_{L}=\set{e\in E\mid e\in D_{i}(u,v), \kappa(D_{i}(u,v))\leqslant \nicefrac{1}{\alpha_{\eps}^2}}$. Similarly, we define $\kappa^+(e)=\kappa(e)\cdot \alpha_{\eps}$. Intuitively, these correspond to the low capacity edges. 
\end{definition2}

\begin{lemma2}\label{lem:sampling2}
Let $G$ be a weighted multigraph with maximum edge weight $W$. Suppose $\textsf{mwm}(G)\geqslant \max\set{\nicefrac{\eps\cdot n}{16\cdot W},\nicefrac{\log n}{\eps^4}}$ Then, with high probability, for all $X\subseteq V$, we have
\begin{align*}
\textsf{mwm}(G_s[X])\leqslant \textsf{mwm}(G[X]\cap E_L,\kappa^+)+\eps\cdot \textsf{mwm}(G)
\end{align*}
\end{lemma2}
\begin{proof}
Consider a fixed $X$, and from now on, we use $H\coloneqq G[X]\cap E_L$ and $H_s\coloneqq G_s[X]\cap E_L$. In order to prove this, we will make use of a primal-dual argument. Note that the problem of finding a $
\textsf{mwm}(H,\kappa^+)$ can be formulated as a linear program as follows.

\begin{equation}
\begin{array}{ll@{}ll}
\text{maximize}  & \displaystyle\sum\limits_{e\in E} w(e)\cdot x(e) &\\
\text{subject to}& \\ \\
\sum_{e\ni v} x(e)&\leqslant 1 &\forall v\in X\\
\sum_{e\in G[B]} x(e)&\leqslant \frac{|B|-1}{2}& \forall B\in X_{\text{odd}}\\
x(e)&\leqslant \kappa^+(e)& \forall e\in E
\end{array}
\end{equation}

The corresponding dual program is as follows:

\begin{equation}
\begin{array}{ll@{}ll}
\text{minimize}  & \displaystyle\sum\limits_{u\in V} y_u +\sum_{B\subseteq X_{\text{odd}}} r(B)\cdot \paren{\frac{\card{B}-1}{2}} + \sum_{e\in E(H)} z(e)\cdot \kappa^{+}(e)&\\
\text{subject to}& \\
&\displaystyle y_{u}+y_{v}+z(e) + \sum\limits_{B: (u,v)\in G[B]} r(B)\geqslant w(e)& \forall e\text{ between }u,v, \forall u,v\in X\\
\end{array}
\end{equation}
Let $f(y,z,r)$ denote the optimal of the dual program, then by strong duality, we know that $\textsf{mwm}(H,\kappa^+) = f(y,z,r)$. Similarly, for any uncapacitated graph we have the same primal and dual programs, except we don't have the third constraint in the primal program, and in the dual program we omit the $z$ variables. 
\begin{equation}
\begin{array}{ll@{}ll}
\text{maximize}  & \displaystyle\sum\limits_{e\in E} w(e)\cdot x'(e) &\\
\text{subject to}& \\ \\
\sum_{e\ni v} x'(e)&\leqslant 1 &\forall v\in X\\
\sum_{e\in G[B]} x'(e)&\leqslant \frac{|B|-1}{2}& \forall B\in X_{\text{odd}}
\end{array}
\end{equation}

The corresponding dual program is as follows:

\begin{equation*}
\begin{array}{ll@{}ll}
\text{minimize}  & \displaystyle\sum\limits_{u\in V} y'_u +\sum_{B\subseteq X_{\text{odd}}} r'(B)\cdot \paren{\frac{\card{B}-1}{2}}&\\
\text{subject to}& &\\
\displaystyle y'_{u}+y'_{v}+ &\sum\limits_{B: (u,v)\in G[B]} r'(B)\geqslant w(e)& \forall e\text{ between }u,v, \forall u,v\in X\\
\end{array}
\end{equation*}

Let $g(y',r')$ denote the optimal for dual program corresponding to $H_s$. Note that this is a random variable, since $H_s$ is a random graph. We will show that with high probability,
\begin{align*}
g(y',r')\leqslant f(y,r,z)+\varepsilon\cdot \textsf{mwm}(G)
\end{align*}
By weak duality, we have, $\textsf{mwm}(H_s)\leqslant g(y',r')$, and $f(y,r,z)=\textsf{mwm}(H,\kappa^+)$. This will show our claim for a fixed $H$, and then we take a union bound over all $H$.\\ \\
We will use $\set{\set{y_{u}}_{u\in X}, \set{r(B)}_{B\in X_{\text{odd}}}}$ to get a solution for the dual program for $H_s$. We will refer to this set of dual variables as an attempted cover. Let $E'=\set{e\mid z(e)>0}$. Observe that the edges left uncovered by attempted cover of $H_s$ are precisely a subset of $E'$. We now modify the cover as follows.
\begin{align*}
\Delta y_u&=\sum_{e\ni v,e\in H_s} z(e)\\
y_u'&=y_u+\Delta y_u
\end{align*}
Note that $\set{\set{y_u'}_{u\in X}, \set{r(B)}_{B\in X_{\text{odd}}}}$ is a valid cover of $H_s$. To see this, consider edge $e\in H_s$ and let $u,v$ be the endpoints of $H_s$. Suppose $e\notin E'$, then, 
\begin{align*}
w(e)&\leqslant y_u+y_v+z(e)+\sum_{e\in G[B]} r(B)\\
&\leqslant y_u+y_v+\sum_{e\in G[B]} r(B)\\
&\leqslant y'_u+y'_v+\sum_{e\in G[B]} r(B)
\end{align*}
Similarly, for an edge $e\in E'$, we have,
\begin{align*}
w(e)&\leqslant y_u+y_v+z(e)+\sum_{e\in G[B]} r(B)\\
&\leqslant y_u+y_v+\Delta y_u+\Delta y_v+\sum_{e\in G[B]} r(B)\\
&\leqslant y'_u+y'_v+\sum_{e\in G[B]} r(B)
\end{align*}
Moreover, $g(y',r')=f(y,r,z)+\sum_{u\in X} \Delta y_u$. Thus, it is sufficient to bound $\sum_{u\in X}\Delta y_u$.
\begin{align*}
\expect{\sum_{u\in X}\Delta y_{u}}&=\sum_{e\in E} z(e)\cdot \prob{e\in H_s}\\
&\leqslant \sum_{e\in E} z(e)\cdot \kappa(e)\cdot \rho_{\eps}\\
						   &\leqslant \sum_{e\in E} \frac{z(e)\cdot \kappa^{+}(e)}{2^{\nicefrac{W}{\eps^2}}}\\
						   &\leqslant \textsf{mwm}(H,\kappa^{+})\cdot 2^{\nicefrac{-W}{\eps^2}}
\end{align*}
Using Chernoff and choosing $\delta = \frac{\eps\cdot \textsf{mwm}(G)\cdot 2^{\nicefrac{W}{\eps^2}}}{\textsf{mwm}(H,\kappa^+)}$, we have,
\begin{align*}
\prob{\sum_{u\in X}\Delta y_u\geqslant 2\cdot \eps\cdot \textsf{mwm}(G)}&\leqslant \prob{\sum_{u\in X}\Delta y_u\geqslant \textsf{mwm}(H,\kappa^{+})\cdot 2^{\nicefrac{-W}{\eps^2}}+\varepsilon\cdot \textsf{mwm}(G)}\\
&\leqslant \exp\paren{\eps\cdot \textsf{mwm}(G) -\eps\cdot \textsf{mwm}(G) \cdot \log \paren{1+\eps\cdot 2^{\nicefrac{W}{\eps^2}}}}\\
&\leqslant \exp\paren{-\eps\cdot \textsf{mwm}(G)\cdot \frac{W}{\eps^2}}
\end{align*}
By assumption, $n\leqslant \frac{\textsf{mwm}(G)\cdot W}{\eps}$, thus taking a union bound over all subsets, we have our theorem.
\end{proof}

\subsection{Phase 3: Finding Set E*}

The algorithm \textsc{WeightedM-or-E*}() proceeds to Phase 3, if $\textsf{mwm}(G,\kappa)$ is not large enough (see \Cref{lem:phase1}). At this point we have to increase capacity along some edges in the graph. Phase 3 finds such edges $E^*$, and makes sure it has that this set has the property that $\sum_{e\in E^*}w(e)\cdot \kappa(e)=O(\textsf{mwm}(G)\log n)$, and moreover, every good matching in $G$ has significant weight going through $E^*$. \\ \\
We will first start with describing the dual program corresponding to the general matching LP:

\begin{equation*}
\begin{array}{ll@{}ll}
\text{minimize}  & \displaystyle\sum\limits_{u\in V} y_u +\sum_{B\subseteq V_{\text{odd}}} r(B)\cdot \paren{\frac{\card{B}-1}{2}} \\
\text{subject to}& \displaystyle y_{u}+y_{v}+ \sum\limits_{B: (u,v)\in G[B]} r(B)\geqslant w(u,v) & & \forall (u,v)\in E 
\end{array}
\end{equation*}
Additionally, we define $yr(e)$ to be the dual constraint corresponding to the edge $e$, and$f(y,r)$ to be the value of the objective function. We now describe some properties of \textsc{Static-Weighted-Match}($H,\eps$). We state these properties without proof for now and postpone the proof to the appendix. 
\begin{restatable}{lemma2}{swm}\label{lem:mwm}\cite{DP14}
There is an $O(\nicefrac{m}{\eps}\cdot \log \nicefrac{1}{\eps})$ time algorithm \textsc{Static-Weighted-Match}() that takes as input, a weighted graph $G$, and a parameter $\eps>0$, and outputs an integral matching $M$, and dual vectors $\vec{y}$ and $\vec{r}$ with the following properties.
\begin{enumerate}
\item It returns an integral matching $M$ such that $w(M)\geqslant (1-\eps)\cdot \textsf{mwm}(G)$ \label{item:mwma}
\item A set $\Omega$ of laminar odd-sized sets such that $\set{B\in V_{\text{odd}}\mid r(B)>0}\subseteq \Omega$. \label{item:mwmb}
\item For all odd-sized sets $B$ such that $\card{B}\geqslant \nicefrac{1}{\eps}+1$, $r(B)=0$. \label{item:mwmc}
\item For all $v\in V$, $y(v)$ is an integral multiple of $\eps$, and for all $B$ with $\card{B}$ odd, $r(B)$ is an integral multiple of $\eps$. \label{item:mwmd}
\item For each edge $e\in E$, we have $yr(e)\geqslant (1-\eps)\cdot w(e)$, that is $e$ is \emph{approximately covered by} $\vec{y}$ and $\vec{r}$. \label{item:mwme}
\item The value of the dual objective, $f(y,r)$ is at most $(1+\eps)\cdot \textsf{mwm}(G)$. \label{item:mwmf}
\end{enumerate}
\end{restatable}

We now state our first claim. 

\begin{claim2}\label{claim:LotofMedges}
Suppose $H\subseteq G$, and $\vec{y}, \vec{r}$ are the dual vectors returned by \textsc{Static-Weighted-Match}($H,\eps$) and define $E_{H}=\set{e\in E\mid yr(e)\geqslant (1-\eps)\cdot w(e)}$. Let $M$ be any matching of $G$, then $w(M\cap E\setminus E_{H})\geqslant w(M)-(1+\eps)^2\cdot\textsf{mwm}(H)$. 
\end{claim2}
\begin{proof}
Observe that if we scale up the dual variables $\vec{y}$ and $\vec{r}$ by $(1+\eps)$, then $\vec{y}$ and $\vec{r}$ is a feasible solution for the dual matching problem for the graph $E_H$. Thus, by weak duality and \Cref{lem:mwm}\ref{item:mwmf}, we have $\textsf{mwm}(E_H)\leqslant (1+\eps)\cdot f(y,r)\leqslant (1+\eps)^2\cdot \textsf{mwm}(H)$. This follows from \Cref{lem:mwm}\ref{item:mwmf}. Thus, we have the following line of reasoning:
\begin{align*}
w(M)&=w(M\cap E_H)+w(M\cap E\setminus E_H)\\
	&\leqslant \textsf{mwm}(E_H)+w(M\cap E\setminus E_H)\\
	&\leqslant (1+\eps)^2\cdot \textsf{mwm}(H)+w(M\cap E\setminus E_H).
\end{align*}
This proves the claim. 
\end{proof}

We now describe the set $E^*$, and show that $w(\kappa(E^*))=O(\textsf{mwm}(G)\cdot \log n)$ with high probability. 

\begin{lemma2}\label{lem:setE*}
Let $G$ be a multigraph such that $\textsf{mwm}(G)\geqslant \nicefrac{\eps n}{W}$, and let $\kappa$ be the capacity function on the edges of the graph. Suppose $G_s$ is the graph obtained by sampling edge $e$ with probability $p(e)=\min\set{1, \rho_{\eps}\cdot \kappa(e)}$. Let $\vec{y},\vec{r}$ be the duals returned by \textsc{Static-Weighted-Match}($G_s,\eps$). Let $E^*=\set{e\mid yr(e)<(1-\eps)\cdot w(e)}$, then with high probability, $w(\kappa(E^*))=O(\textsf{mwm}(G)\cdot \log n)$. 
\end{lemma2}

\begin{proof}
We consider the set $\mathcal{D}$ of duals $\vec{y},\vec{r}$ that satisfy the following properties.
\begin{enumerate}
\item For all $v\in V$, $y(v)$ is a multiple of $\eps$, and for all $B\in V_{\text{odd}}$, $r(B)$ is a multiple of $\eps$. \label{item:a}
\item Let $\Omega=\set{B\mid r(B)>0}$, then $\Omega$ is laminar. \label{item:b}
\item If $r(B)>0$ for some $B$, then $\card{B}\leqslant \nicefrac{1}{\eps}$. \label{item:c}
\end{enumerate}
Observe that $\mathcal{D}$ contains all possible duals that could be returned by \textsc{Static-Weighted-Match}(). Now, we bound $\card{\mathcal{D}}$.
\begin{align*}
\card{\mathcal{D}}&\leqslant \sum_{i=0}^{2n} {n^{\nicefrac{1}{\eps}} \choose i}\cdot \paren{\frac{W}{\eps}}^{n}\cdot \paren{\frac{W}{\eps}}^{2n}\\
				&\leqslant 2n\cdot {n^{\nicefrac{1}{\eps}} \choose n}\cdot \paren{\frac{W}{\eps}}^{n}\cdot \paren{\frac{W}{\eps}}^{2n}\\
				&\leqslant 2n\cdot {n^{\nicefrac{1}{\eps}} \choose n}\cdot \paren{\frac{n^2}{\eps^3}}^{n}\cdot \paren{\frac{n^2}{\eps^3}}^{2n}\\
				&\leqslant 2^{\nicefrac{10}{\eps^3}\cdot n\cdot \log n}\\
				&\leqslant 2^{\nicefrac{10}{\eps^3}\cdot \frac{\textsf{mwm}(G)}{W}\cdot \log n}
\end{align*}
This follows from the following argument. Since $\Omega$ is laminar (due to \ref{item:b}), there are at most $2n$ sets contained in $\Omega$, and from \ref{item:c}, we can conclude that these laminar sets are chosen from among $n^{\nicefrac{1}{\eps}}$ sets. Similarly, \ref{item:a} suggests that each $y(v)$ and chosen $r(B)$ can be assigned $\nicefrac{W}{\eps}$ values, thus for a given choice of $\Omega$, there are at most $\paren{\frac{W}{\eps}}^{n}\cdot \paren{\frac{W}{\eps}}^{2n}$ choices for $\vec{y}$ and $\vec{r}$. We can derive the last set of equations by observing that $W\leqslant \nicefrac{n^2}{\eps^2}$ (see \Cref{obs:ratiominmaxweight})and by the premise of our lemma, we have that $\textsf{mwm}(G)\geqslant \nicefrac{\eps n}{W}$. Additionally, observe that $\mathcal{D}$ includes the set of duals that can be returned by $G$, and therefore $\card{\mathcal{D}}$ is an upper bound on the set of possible duals returned by \textsc{Static-Weighted-Match}(). \\ \\ 
Now, consider any $\vec{y}, \vec{r}\in \mathcal{D}$, and let $E_{y,r}=\set{e\mid yr(e)<(1-\eps)\cdot w(e)}$. Observe that for all $e\in E_{y,r}$, we have $\kappa(e)<\nicefrac{1}{\rho_{\eps}}$, since otherwise $\kappa(e)=1$, and $e$ would be included in $G_s$ with probability $1$. Thus, for all $e\in E_{y,r}$, we have $\kappa(e)<\nicefrac{1}{\rho_{\eps}}$. Now, $\vec{y},\vec{r}$ are such that $w(\kappa(E_{y,r}))>\textsf{mwm}(G)\cdot \log n$, then none of the edges in $E_{y,r}$ were sampled. Note that in this case, $\kappa(E_{y,r})\geqslant \textsf{mwm}(G)\cdot \log n\cdot W^{-1}$. Let $\mathcal{E}_1$ denote the event that $\vec{y},\vec{r}$ are returned by \textsc{Static-Weighted-Match}($G_s,\eps$) and $\mathcal{E}_2$ be the event that $E_{y,r}$ is not sampled into $G_s$. Then, observe that in this case, 
\begin{align*}
\prob{\mathcal{E}_1}&\leqslant \prob{\mathcal{E}_2}\\
                                 &\leqslant \prod_{e\in E_{y,r}}(1-p(e))\\
                                 &\leqslant \exp\paren{-\sum_{e\in E_{y,r}}\kappa(e)\cdot \rho_{\eps}}\\
                                 &\leqslant \exp\paren{-\frac{ \textsf{mwm}(G)}{W}\cdot \rho_{\eps}}
\end{align*}
Now, in order to upper bound this, we take a union bound over all $\card{\mathcal{D}}$ many possible duals. From the previous discussion these we know that $\card{\mathcal{D}}\leqslant \exp\paren{\frac{10}{\eps^2}\cdot \frac{\textsf{mwm}(G)}{W}\cdot \log n}$. This concludes the proof.
 \end{proof}

\subsection{Algorithm for \textsc{Weighted-Frac-Match}()}

The final ingredient we need for \textsc{WeightedM-or-E*}() is the following lemma, which computes a fractional matching on low capacity edges.

\begin{lemma2}\label{lem:Hungarian2}
Consider a weighted bipartite multigraph $G$ with edge capacity function $\kappa$, and edge weights in $\set{1,2,\cdots, W}$. There is an algorithm that in $O(m\cdot W\cdot \log n\cdot \nicefrac{1}{\eps})$ time finds a fractional matching $\vec{x}$ obeying the capacity function $\kappa$ such that $\sum_{e\in E}w(e)\cdot x(e)\geqslant (1-\eps)\cdot \textsf{mwm}(G,\kappa)$.
\end{lemma2}

In order to come up with this algorithm, we first recall the maximum weight capacitated fractional matching linear program for \textbf{bipartite graphs}, and its corresponding dual.

\begin{equation}
\begin{array}{ll@{}ll}
\text{maximize}  & \displaystyle\sum\limits_{e\in E} w(e)\cdot x(e)  \\
\text{subject to}& \displaystyle x(e)\leqslant \kappa(e)& & \forall e\in E \\
& \displaystyle \sum_{e\ni v} x(e) \leqslant 1 & & \forall v\in V
\end{array}
\end{equation}

\begin{equation}
\begin{array}{ll@{}ll}
\text{minimize}  & \displaystyle\sum\limits_{v\in L} y(v) +\sum\limits_{u\in R} y(u) +\sum_{e\in E} z(e)\cdot \kappa(e) &  \\
\text{subject to}& \displaystyle yz(e)\geqslant w(e) &  \forall e\in E 
\end{array}
\end{equation}

Here, $yz(e)=y(u)+y(v)+z(e)$, where $u,v$ are the two endpoints of $e$. Additionally, before stating our algorithm, we give some definitions.

\begin{definition2}[Residual Graph]
Given a bipartite graph $G=(L,R,E)$, with a capacity function $\kappa$ on the edges of the graph. Let $\vec{x}$ be a fractional matching obeying $\kappa$. We define $G_x$ to be the residual graph with respect to $\vec{x}$. In particular, this is a directed graph, where corresponding to each undirected edge $e\in G$, there are two directed edges, $e_f$ (forward edges, directed from $L$ to $R$), and $e_b$ (backward edge, directed from $R$ to $L$). Moreover, $e_f$ has a residual capacity of $\kappa(e)-x(e)$, and $e_b$ has a capacity of $x(e)$. Thus, if an edge $e$ is saturated, that is $x(e)=\kappa(e)$, then $e_f$ is not in $G_x$. Similarly, if $x(e)=0$, then $e_b$ is not in $G_x$.  
\end{definition2}

\begin{definition2}[Free Vertices]
Given a graph $G$, and a fractional matching $\vec{x}$, we say that a vertex $v$ is free with respect to $\vec{x}$ if $\sum_{e\ni v}x(e)<1$.
\end{definition2}

\begin{definition2}
    An augmenting path in $G_x$, will refer to a path starting and ending in a free vertex, and all the intermediate vertices on this path will be saturated.
\end{definition2}

Before stating our algorithm, we will state the invariants that the algorithm will maintain. Subsequently, we will show that maintaining these invariants throughout implies that the algorithm will output a matching satisfying the approximation guarantees of \Cref{lem:Hungarian2} and in the desired runtime. 

\begin{remark2}
We round down $\eps$ so that $\nicefrac{1}{\eps}$ is an integer.
\end{remark2}

\begin{property}\label{prop:invariants}
Our algorithm will at all times maintain, a fractional matching $\vec{x}$, and the dual variables $y(u)$ for all $u\in V$, and $z(e)$, for all $e\in E$ with the following properties.
\begin{enumerate}
    \item \textbf{Granularity:} Throughout the algorithm, the duals $y(u)$ for all $u\in V$, and $z(e)$ for all $e\in E$, are multiples of $\eps$. \label{item:gran}
    \item \textbf{Domination:} For all directed edges $e_f$ and $e_b$ in $G_{x}$, we have $yz(e_f)\geqslant w(e)-\eps$, and $yz(e_b)\geqslant w(e)-\eps$.\label{item:domin}
    \item \textbf{Tightness:} For all backward edges $e_b$, we additionally have the property that $yz(e_b)\leqslant w(e_b)+\eps$. \label{item:tightness}
    \item \textbf{Free Duals:} The duals for free vertices $u\in L$ are equal, and are at most the dual values for other vertices. The free vertex duals in $R$ are always $0$. The algorithm terminates when the free vertex duals in $L$ are all $0$. \label{item:freeduals}
    \item \textbf{Complementary Slackness:} If $z(e)>0$ for some $e\in E$, then $x(e)=\kappa(e)$. \label{item:slackness}
\end{enumerate}
\end{property}

\begin{definition2}[Eligible Edges]
    Let $\vec{x}$ be the current fractional matching maintained by the graph, and let $G_{x}$ denote the residual graph with respect to $\vec{x}$. A forward edge $e_f$ is eligible if $yz(e_f) = w(e)-\eps$, and backward edge $e_b$ is said to be eligible if $yz(e_b)=w(e)+\eps$. We use $G^t_x$ to denote the eligible subgraph of $G_x$. 
\end{definition2}

\begin{lemma2}
The fractional matching $\vec{x}$ returned by the algorithm mentioned in \Cref{prop:invariants} has $\sum_{e\in E}w(e)\cdot x(e)\geqslant (1-\eps)\cdot \textsf{mwm}(G)$.
\end{lemma2}

\begin{proof}
    In order to see this, consider the dual variables returned by the algorithm. Note that if we increase each of the duals by a factor of $(1+\eps)$, then the dual solutions $\vec{y},\vec{z}$ maintained by the algorithm are feasible dual solutions to the dual program (this is by \Cref{prop:invariants}\ref{item:domin}). Thus, by weak duality, we have,
    \begin{align*}
        \sum_{u\in V}y(u)+\sum_{e\in E} z(e)\cdot \kappa(e)\geqslant (1-\eps)\cdot \textsf{mwm}(G)
    \end{align*}
    Due to complementary slackness, (\Cref{prop:invariants}\ref{item:slackness}
    ), we have,
    \begin{align*}
        \sum_{u\in V} y(u) +\sum_{e\in E} x(e)\cdot z(e) \geqslant (1-\eps)\cdot \textsf{mwm}(G)
    \end{align*}
Due to the fact that free vertex duals, at the end of the algorithm are zero (\Cref{prop:invariants}\ref{item:freeduals}), we have, 
\begin{align*}
    \sum_{u\in V}\sum_{e\ni u} x(e)\cdot y(u)+\sum_{e\in E} z(e)\cdot x(e)\geqslant (1-\eps)\cdot \textsf{mwm}(G)
\end{align*}
Simplifying the above expression, we have,
\begin{align*}
    \sum_{e\in E}yz(e)\cdot x(e)\geqslant (1-\eps)\cdot \textsf{mwm}(G)
\end{align*}
Using \Cref{prop:invariants}\ref{item:tightness}, that is, tightness, we have, 
\begin{align*}
    \sum_{e\in E} (1+4\eps)\cdot w(e)\cdot x(e)\geqslant (1-\eps)\cdot \textsf{mwm}(G)
\end{align*}
This shows the claim.
\end{proof}

We now want to state the precise algorithm.
\begin{algorithm}[H] 
	\algorithmicrequire{ A bipartite multigraph $G$ with capacity and weight edge functions $\kappa$ and $w$.}\\
	\algorithmicensure{ A fractional matching $\vec{x}$ with $\sum_{e\in E}w(e)\cdot x(e)\geqslant (1-\eps)\cdot \textsf{mwm}(G,\kappa)$}
	\caption{\textsc{Weighted-Frac-Matching}($G,\kappa,\eps$)}
	\begin{algorithmic}[1]
	\Procedure{Initialization:}{}
		\State $x(e)\leftarrow 0$ for all $e\in E$ \Comment{Initializing $\vec{x}$}
		\State $y(u)\leftarrow W-\eps$ for all $u\in L$ \Comment{Initializing left vertex duals.}
		\State $y(u)\leftarrow 0$ for all $u\in R$ \Comment{Initializing right vertex duals.}
		\State $z(e)\leftarrow 0$ for all $e\in E$ \Comment{Initializing edge duals.}
	\EndProcedure
		\While{$y$-values of free left vertices are greater than $0$}
		\For{$e_b\in G^t_x$ with $z(e)>0$}
		\State $z(e)\leftarrow z(e)-\min\set{z(e), yz(e)-w(e)+\eps}$. \Comment{Either $z(e)=0$, or $e\notin G^t_x$ as $yz(e)=w(e)-\eps$} \label{item:zadj} 
		\EndFor
		\State Update $G^t_x$
		\Procedure{Augmentation:}{} 
		\State Find the maximal set $\mathcal{P}$ of augmenting paths in $G^t_x$. \label{item:augpath}
		\State Augment along $\mathcal{P}$. Update $\vec{x}$ and $G^t_x$.  \label{item:augmentation}
		\EndProcedure
	\Procedure{Dual Adjustment:}{}
		\State Let $Z$ be the set of vertices reachable from free left vertices in $G^t_x$.
		\State For all ineligible edges $e_b$ directed from $R\setminus Z$ to $L\cap Z$, with $yz(e)=w(e)-\eps$, $z(e)\leftarrow z(e)+\eps$. \label{line:adjz}
		\State For all $u\in L\cap Z$, $y(u)\leftarrow y(u)-\eps$. \label{line:Ladj}
		\State For all $u\in R\cap Z$, $y(u)\leftarrow y(u)+\eps$.\label{line:Radj}
	\EndProcedure
		\EndWhile
		\end{algorithmic}
	 \label{alg:weightedfracmatching}
\end{algorithm}

We now show that \Cref{alg:weightedfracmatching} proves \Cref{prop:invariants}. Later, we will show runtime guarantees of \Cref{alg:weightedfracmatching}. Towards this, we start with the following observation.

\begin{observation2}\label{obs:simul}
    Let $e\in G$ be an edge, then $e_f$ and $e_b$ cannot simultaneously appear in $G^t_x$
\end{observation2}

\begin{definition2}
    Let $H$ and $G$ be directed, capacitated graphs with capacity functions $c_h$ and $c_g$ respectively. Then, we say $E(H)\subseteq E(G)$ if $e\in H$ implies $e\in G$, and $c_h(e)\leqslant c_g(e)$. 
\end{definition2}

\begin{claim2}\label{claim:noaugpaths}
After \Cref{item:augmentation}, there are no augmenting paths in the graph $G^t_x$. 
\end{claim2}
\begin{proof}
Let $\vec{y}$ denote the fractional matching after the augmentation step. Then, in order to show the claim, it is sufficient to show $E(G^t_y)\subseteq E(G^t_x)$. Consider any forward edge $e_f\in E(G^t_x)$, observe that $e_b\notin E(G^t_x)$ due to \Cref{obs:simul}. Since we don't change the duals, $e_b\notin E(G^t_y)$ as well. Finally, $e_f$ can either be augmented along or not, in either case, we have, $\kappa(e)-x(e)\geqslant \kappa(e)-y(e)$. Such an argument applies for any edge $e_b$ as well.\\ \\
From this, we can conclude that any augmenting path that is present in $G^t_y$ is in $G^t_x$ as well, and this contradicts the fact that we found a maximal set of augmenting paths $\mathcal{P}$.
\end{proof}

We now show that \textbf{granularity} property holds.
\begin{claim2}\label{proof:gran}
    Throughout the algorithm, \Cref{prop:invariants}\ref{item:gran} holds.
\end{claim2}
\begin{proof}
We show this property by induction. Note that at the beginning of the algorithm, this property holds since $y(u)=W-\eps$ for all $u\in L$. Since $\nicefrac{1}{\eps}$, and $W$ are integers, this claim holds for all $y(u)$ with $u\in L$. Moreover, $y(u)=0$ for all $u\in R$, and $z(e)=0$ for all $e\in E$, thus, these are also multiples of $\eps$. Next, look at each of the steps where the duals are modified. Consider the step \Cref{item:zadj}, if $z(e)$ was a multiple of $\eps$ before this step, then it is a multiple of $\eps$ after this step as well since $w(e)$ is an integer. \\ \\
Finally, consider the \textsc{Dual Adjustment} procedure. If the duals are multiples of $\eps$ initially, then they are multiples of $\eps$ even after this step, since we modify the duals by $\eps$ amount. 
\end{proof}
We now show that \textbf{domination} property holds.
\begin{claim2}
Throughout the algorithm, \Cref{prop:invariants}\ref{item:domin} holds.
\end{claim2}
\begin{proof}
We show this property by induction. Note that initially, for all $u\in L$, $y(u)=W-\eps$, so domination condition holds for all edges initially. Now consider the step \Cref{item:zadj}, after this step, $yz(e)\geqslant w(e)-\eps$, which satisfies domination for $e$, and since we only modified $z(e)$, this doesn't affect any other edges. \\ \\
Finally, consider the dual adjustment procedure. Consider $e_f$ between $L\cap Z$ and $R\setminus Z$. Note that $e_f$ has to be ineligible, thus $yz(e)>w(e)-\eps$, and due to granularity, we can deduce that $yz(e)\geqslant w(e)$. Thus, the dual adjustment step preserves domination for such edges. Next, consider a backward edge $e_b$ between $L\cap Z$ and $R\setminus Z$. If $yz(e)\geqslant w(e)$, then dual adjustment step preserves domination. On the other hand, if $yz(e)=w(e)-\eps$, then we increment $z(e)$ as well, therefore domination is still preserved. Consider edges between $R\cap Z$ and $L\setminus Z$, since duals of $R\cap Z$ are incremented, and the duals of $L\setminus Z$ are unchanged, the domination property is preserved for such edges. Edges between $L\cap Z$ and $R\cap Z$ are not affected. 
\end{proof}

We now show the \textbf{tightness} property. 

\begin{claim2}
Throughout the algorithm, \Cref{prop:invariants}\ref{item:tightness} property holds.
\end{claim2}
\begin{proof}
    We prove this by induction. At the start of the algorithm, there are no backward edges, so tightness property holds trivially. Consider \Cref{item:zadj}, this step only reduces the dual, so it can't affect the tightness property. Consider the dual adjustment procedure. Here, \Cref{line:adjz} can affect the tightness property. However, this line only applies to backward edges $e_b$ which have $yz(e)=w(e)-\eps$, thus after dual adjustment $yz(e)=w(e)$. \Cref{line:Radj} can affect backward edges between $L\setminus Z$ and $R\cap Z$. However, all such edges $e_b$ must be ineligible. Thus, $yz(e)<w(e)+\eps$, which by granularity implies that $yz(e)\leqslant w(e)$. Thus, after the dual adjustment step, $yz(e)=w(e)+\eps$.  
\end{proof}

We now show the \textbf{free duals} property.

\begin{observation2}
    Throughout the algorithm, \Cref{prop:invariants}\ref{item:freeduals} holds.
\end{observation2}
\begin{proof}
    This is evident from the fact that all free vertices of $L$ are contained in $Z$, and duals of all these vertices is decremented by the same amount. Similarly, \Cref{claim:noaugpaths} suggests that all free vertices of $R$ are contained in $R\setminus Z$. Since, free vertex duals on the right are initially $0$, and are never incremented since they are contained in $R\setminus Z$. Thus, they remain $0$ throughout. 
\end{proof}

Finally, we show \textbf{complementary slackness}.

\begin{claim2}
Throughout the algorithm, \Cref{prop:invariants}\ref{item:slackness} holds.
\end{claim2}
\begin{proof}
    We will again show this by induction. This property holds at the beginning of the algorithm since $z(e)=0$ for all $e\in E$. Observe that this property may be violated if $z(e)>0$ and $e_b$ is present in $G^t_x$, and $e_b$ is used in some augmenting path. However, \Cref{item:zadj} ensures this doesn't happen. On the other hand, in \Cref{line:adjz}, we are incrementing $z(e)$, and this may violate complementary slackness as well. However, we show that for such an edge $x(e)=\kappa(e)$. To see this, note that $yz(e)=w(e)-\eps$, so if $x(e)<\kappa(e)$, then $e_f\in G^t_x$. But this cannot happen since we can't haven any eligible forward edges from $L\cap Z$ to $R\setminus Z$. Thus, this step doesn't contradict complementary slackness.
\end{proof}

Now, we focus on the runtime analysis.

\subsection{Runtime}
To prove the guarantees on the runtime of \Cref{alg:weightedfracmatching}, we start with the following structural lemma about $G^t_x$

\begin{observation2}\label{obs:addingedges}
    After dual adjustment step, there are only forward edges between $L\cap Z$ and $R\setminus Z$, and backward edges between $R\cap Z$ and $L\setminus Z$.
\end{observation2}
\begin{proof}
Suppose there is a backward edge $e_b$ between $L\cap Z$ and $R\setminus Z$ after the dual adjustment step. This implies, that after the dual adjustment step, $yz(e_b)=w(e)+\eps$. On the other hand, prior to the dual adjustment step, we had that $yz(e_b)=w(e)+2\eps$. However, this would contradict the tightness property (see \Cref{prop:invariants}\ref{item:tightness}). Similarly, suppose there is a forward edge $e_f$ between $R\cap Z$ and $L\setminus Z$ after the dual adjustment step. Then, this would imply that $yz(e_f)=w(e)-\eps$. Thus, before the dual adjustment step, $yz(e_f)=w(e)-2\eps$, which would contradict the domination property (see \Cref{prop:invariants}\ref{item:domin}).
\end{proof}

\begin{lemma2}\label{lem:acyclic}
The graph $G^t_x$ that is fed into \textsc{Augmentation} procedure is acyclic.
\end{lemma2}
\begin{proof}
We will show this by induction. Observe that at the beginning of the algorithm, there are no back edges, so the graph is acyclic. We now show that if no cycles are present in the graph, then dual adjustment and augmentation procedures cannot create cycles.
\begin{enumerate}
    \item First consider \Cref{item:zadj}, observe that if $e_b\in G^t_x$ is subjected to this, then $x(e)=\kappa(e)$ due to complementary slackness. Thus $e_f$ does not exist, and cannot be added to $G^t_x$ due to this modification. Moreover, this modification does not affect any other edge apart from $e$. Thus, if the graph $G^t_x$ was acyclic before, then it remains acyclic after \Cref{item:zadj}.
    \item Let $\vec{y}$ be the fractional matching after augmentation step. As we saw in \Cref{claim:noaugpaths}, $E(G^t_y)\subseteq E(G^t_x)$. Thus, if $G^t_x$ is acyclic, then $G^t_y$ is acyclic as well.
    \item Now consider \Cref{line:adjz}. By induction hypothesis, there are no cycles in $G^t_x$ before \Cref{line:adjz} was executed. Suppose a cycle $C$ exists in $G^t_x$ after this step, then $C$ contains either $e_f$ or $e_b$, where $e$ is an edge subjected to \Cref{line:adjz}. However, for $e_b$ subjected to this, $\kappa(e)=x(e)$, and therefore, $e_f$ cannot exist in $G^t_x$. Moreover, $e_b\notin G^t_x$ since before the dual modification $yz(e)=w(e)-\eps$, and after dual adjustment, $yz(e)=w(e)$.
    \item Observe that by induction hypothesis, there are no cycles in $G^t_x$ before \Cref{line:Ladj} and \Cref{line:Radj}. Note that these steps preserve eligibility of any edge between $L\cap Z$ and $R\cap Z$ and those between $L\setminus Z$ and $R\setminus Z$. Thus, if a new cycle is created after \Cref{line:Ladj} and \Cref{line:Radj}, then it must contain a directed edge from $Z$ to $V\setminus Z$, and a directed edge from $V\setminus Z$ to $Z$. However, \Cref{obs:addingedges} suggests that after this step, the latter edges are not in $G^t_x$. This shows our claim.
\end{enumerate}
\end{proof}

We now proceed with runtime guarantees.

\begin{observation2}\label{obs:iterations}
    Observe that by definition of augmenting paths, a matched vertex is never made free. Thus, by step \Cref{line:Ladj}, at the end of each iteration of the while loop, the value of free duals in $L\cap Z$ drops by $\eps$. This implies that the total number of iterations is $O(\frac{W}{\eps})$.
\end{observation2}

\begin{lemma2}\cite{SleatorT83}\label{lem:findingaugpaths}
If the residual graph is acyclic, then we can find a maximal set of augmenting paths in time $O(m\log n)$.
\end{lemma2}

\begin{lemma2}
The total runtime of the algorithm is $O(\nicefrac{m}{\eps}\cdot W\cdot \log n)$.
\end{lemma2}
\begin{proof}
    In a particular iteration of a while loop, we see how each of the steps contribute to the runtime. Consider step \Cref{item:zadj}, this takes time $O(m)$ and similarly, the subsequent step of updating $G^t_x$ 
    takes time $O(m)$ as well. Since $G^t_x$ is acyclic (by \Cref{lem:acyclic}), we can use \Cref{lem:findingaugpaths} to find the maximal set of augmenting paths in time $O(m\cdot \log n)$. Finally, the dual adjustment steps can also be accomplished in time $O(m)$. From \Cref{obs:iterations} implies the total runtime is $O(\nicefrac{m}{\eps}\cdot W\cdot \log n)$.
\end{proof}
 
\section{Algorithm \textsc{WeightedM-or-E*()}}
 
 In this section, we now put everything together and describe the algorithm \textsc{WeightedM-or-E*}(). First, some definitions are in order. 
 
 \begin{definition2}\label{def:lowcapedges}
 Given a matching $M$, define $E^{M}_{L}(G,\kappa)=E_{L}(G,\kappa)\cap M$, and $V^{M}_{L}$ to be the set of vertices that are the endpoints of $E^{M}_{L}(G,\kappa)$. 
 \end{definition2}

\begin{definition2}
    Given a multigraph $G$ with weight function $w$ and capacity function $\kappa$ on the edges of $G$, we define the bipartite double cover of $G$, denoted $\textsc{bc}(G)$, with capacity function $\kappa_{\textsc{bc}}$ and weight function $w_{\textsc{bc}}$ as follows:
    \begin{enumerate}
        \item For every vertex $v\in V(G)$, make two copies $v$ and $v'$ in $V(\textsc{bc}(G))$. 
        \item If $e$ is an edge between $u,v$, then we add an edge $e'$ between $u$ and $v'$ and an edge $e''$ between $u'$ and $v$. Moreover, $w_{\textsc{bc}}(e')=w_{\textsc{bc}}(e'')=w(e)$ and $\kappa_{\textsc{bc}}(e')=\kappa_{\textsc{bc}}(e'')=\kappa(e)$.
    \end{enumerate}
\end{definition2}

\begin{observation2}\label{lem:matchingbc}
For any weighted graph $G$, with capacity function $\kappa$, we have, $\textsf{mwm}(\textsc{bc}(G),\kappa_{\textsc{bc}})\geqslant 2\cdot \textsf{mwm}(G,\kappa)$.
\end{observation2}

 
 \begin{lemma2}\label{lem:Hungarian}
 Given a weighted multigraph $G$ (possibly non-bipartite) with capacity function $\kappa$, $\eps\in (0,1)$, with edge weights from $\set{1,2,\cdots, W}$, and with the property that for all $i\in [W]$, and for all vertices $u,v\in V$ we have, $\kappa(D_{i}((u,v)))\leqslant \nicefrac{1}{\alpha_{\eps}}$, then there is an algorithm \textsc{Weighted-Frac-Match-General}() that takes as input $G, \kappa, w,\eps$ and outputs a fractional matching $\vec{x}$ such that $\sum_{e\in E}w(e)\cdot x(e)\geqslant (1-\eps)\cdot \textsf{mwm}(G,\kappa)$, and further $\vec{x}$ obeys the capacities and the odd set constraints. 
 \end{lemma2}
 \begin{proof}
     In the previous section, we saw the algorithm \textsc{Weighted-Frac-Match}() that solved this problem for bipartite graphs. For the case of general graphs, we proceed as follows. We run \textsc{Weighted-Frac-Match}($\textsc{bc}(G),\kappa_{\textsc{bc}},w_{\textsc{bc}},\eps$), and obtain a fractional matching $\vec{z}$. To translate this into a valid fractional matching $\vec{x}$ in $G$, do the following: for each $e\in G$, consider $e',e''\in \textsc{bc}(G)$, and let $x(e)=\frac{z(e')+z(e'')}{2}$. Note that $\vec{z}$ is a fractional matching, since $\vec{x}$ is and moreover since $\vec{x}$ satisfies capacity constraints since $\vec{x}$ does. Applying \Cref{lem:matchingbc}, we know that $\sum_{e\in E}w(e)\cdot x(e)\geqslant (1-\eps)\cdot \textsf{mwm}(G,\kappa)$. Moreover, $\frac{\vec{x}}{1+\eps}$ satisfies all odd set constraints since it satisfies the premise of \Cref{obs:litteflow}.
 \end{proof}
 
 We now state our main algorithm.

\begin{algorithm}[H]
	\caption{\textsc{WeightedM-or-E*}($G,\kappa,\eps,\mu,w$)}
	\begin{algorithmic}[1]
		\State Include each edge $e\in E(G)$ independently with probability $p(e)=\kappa(e)\cdot \rho_{\eps}$ into graph $G_s$.
		\State Let $M,\vec{y},\vec{r}$ be the output of \textsc{Static-Weighted-Match}($G_s,\eps$). \Comment{Phase 1}.
		\If{$w(M)\leqslant (1-6\eps)\cdot\textsf{mwm}(G)$}
		\State Return $E^*=\set{e\mid yr(e)<(1-\eps)\cdot w(e)}$. \Comment{Phase 3}
		\Else \Comment{Phase 2}
		\State $M_I\leftarrow M\setminus E_{L}(G,\kappa)(G,\kappa)$\label{line:defineMI}
		\State $\vec{y}\leftarrow M_I^{D}$\Comment{Convert $M_I$ into a matching on multigraph}
		\State $\vec{x}\leftarrow \textsc{Weighted-Frac-Match-General}(G[V^M_L]\cap E_{L}(G,\kappa),\kappa^+, \eps,w)$ \label{item:wfm}
		\EndIf
		\State Return $\vec{y}+\vec{x}$. 
		\end{algorithmic}
	\label{alg:WeightedMorE}
\end{algorithm}

We restate \Cref{lem:MorE}, then proceed to show that \Cref{alg:WeightedMorE} satisfies the conditions of the lemma. 

\MorE*

\begin{proof}
We first show the runtime of the algorithm. First note that $G_s$ can be computed in $O(m)$ time since this only involves sampling edges independently. Moreover, by \Cref{lem:mwm}, we know that the runtime of \textsc{Static-Weighted-Match}($G,\eps$) is $O(\nicefrac{m}{\eps}\cdot \log \nicefrac{1}{\eps})$. Additionally, from \Cref{lem:setE*}, we can conclude that we can obtain set $E^*$ using the output of \textsc{Static-Weighted-Match}($G_s,\eps$) and the time taken to do this is $O(m)$. Finally, from \Cref{lem:Hungarian}, we can conclude that we can run Line \ref{item:wfm} in $O(\nicefrac{m}{\eps}\cdot \log\nicefrac{1}{\eps}\cdot W\cdot \log n)$. \\ \\
Now, we turn to show \Cref{lem:MorE}\ref{item:MorEa}. In this case, first note that $V_{L}^M$ and $V(M_I)$ are disjoint. This follows from \Cref{line:defineMI} and \Cref{def:lowcapedges}. Thus, since $\vec{y}, \vec{x}$ are fractional matchings, we can conclude that $\vec{z}$ is also a fractional matching.  Note that we land in \Cref{lem:MorE}\ref{item:MorEa} if $w(M)\geqslant (1-\eps)\cdot \textsf{mwm}(G)$. From \Cref{lem:sampling2} we can conclude that for $H=E_{L}(G,\kappa)\cap G[V_{L}^M]$ and $H_s=E_{L}(G,\kappa)\cap G_s[V_{L}^M]$, we have that  $\textsf{mwm}(H_s)\leqslant \textsf{mwm}(H,\kappa^+)+\eps\cdot \textsf{mwm}(G) $. Thus, we have,
\begin{align*}
\sum_{e\in E}y(e)\cdot w(e)+\sum_{e\in E} x(e)\cdot w(e)&\geqslant w(M_I)+ \textsf{mwm}(H,\kappa^+)-\eps\cdot \textsf{mwm}(G)\\
&\geqslant w(M_I)+ \textsf{mwm}(H_s)-2\eps\cdot \textsf{mwm}(G)\\
&=w(M)-2\eps\cdot \textsf{mwm}(G)\\
&\geqslant(1-\eps)\cdot \textsf{mwm}(G)-2\eps\cdot \textsf{mwm}(G)
\end{align*}
We now proceed to proof \Cref{lem:MorE}\ref{item:MorEa}\ref{item:MorEai} and \ref{item:MorEaii}. Consider any edge $e\in \text{supp}(\vec{z})$ with $w(e)=i$ and $\kappa(D_{i}(e))\leqslant \nicefrac{1}{\alpha_{\eps}^2}$.  Thus, we know that $e\in \text{supp}(\vec{x})$, and so by \Cref{lem:Hungarian}, we have that $z(e)=x(e)\leqslant \kappa^{+}(e)\leqslant \kappa(e)\cdot \alpha_{\eps}$ and $z(D_i(e))=x(D_i(e))\leqslant \kappa^{+}(D_i(e))\leqslant \kappa(D_{i}(e))\cdot \alpha_{\eps}$. Similarly, for any $e\in \text{supp}(\vec{z})$ with  $w(e)=i$, and $\kappa(D_i(e))>\nicefrac{1}{\alpha_{\eps}^2}$. We know that $e\in \textsf{supp}(\vec{y})$. By definition of $\vec{y}$, we have $z(e)=y(e)=\nicefrac{\kappa(e)}{\kappa(D_i(e))}$, and $z(D_i(e))=1$, this proves our claim.\\ \\
Finally, we prove \Cref{lem:MorE}\ref{item:MorEb}. First consider \Cref{lem:setE*}, this lemma implies that $\sum_{e\in E^*}\kappa(e)\cdot w(e)=O(\textsf{mwm}(G)\cdot \log n)$. Finally, by \Cref{claim:LotofMedges}, we have, $\sum_{e\in M\cap E^*} w(e) \geqslant w(M)-(1+\eps)^2\cdot \textsf{mwm}(G_s)$, thus, for any $M$ with $w(M)\geqslant (1-\eps)\cdot \textsf{mwm}(G)$, we have $\sum_{e\in M\cap E^*}w(e)\geqslant \eps\cdot \textsf{mwm}(G)$. This is implied by the fact that the algorithm returns $E^*$ when $\textsf{mwm}(G_s)\leqslant (1-5\eps)\cdot \textsf{mwm}(G)$. Finally, for all edges in $E^*$, $\kappa(e)<1$, otherwise they'd be sampled into $G_s$. 
\end{proof}

\section*{Acknowledgement}
Thank you to Jiale Chen, Aaron Sidford, and Ta-Wei Tu for coordinating posting to arXiv.\\ \\
This work is supported by Austrian
Science Fund (FWF): P 32863-N. This project has received funding from the European Research Council (ERC) under the
European Union’s Horizon 2020 research and innovation programme (grant agreement No 947702).

\appendix

\section{Missing Proofs}

We start with the following observation:
\begin{observation2}\label{obs:smallodd}
Suppose $\vec{x}$ is a fractional matching that satisfies odd set constraints for all odd sets of size smaller than $\nicefrac{3}{\eps}+1$, then the fractional matching $\vec{z}=\frac{\vec{x}}{(1+\eps)}$ satisfies odd set constraints for all odd sets. 
\end{observation2}
\begin{proof}
    Suppose $\vec{x}$ is a fractional matching that satisfies odd set constraints for all odd sets of size at most $\nicefrac{3}{\eps}$. Let $\vec{z}=\frac{\vec{x}}{1+\varepsilon}$. Now consider any odd set $B$ such that $\card{B}\geqslant \nicefrac{3}{\eps}+1$. Then, we have,
    \begin{align}\label{eqn:blossom}
    \frac{\card{B}}{\card{B}-1}=1+\frac{1}{\card{B}-1}\leqslant 1+\nicefrac{\eps}{3}.
    \end{align}
    Now since $\vec{x}$ satisfies fractional matching constraints, we have, $\sum_{e\in G[B]}x(e)\leqslant \nicefrac{\card{B}}{2}$. Thus, we have, using \Cref{eqn:blossom}:
    \begin{align*}
        \sum_{e\in G[B]}z(e)\leqslant \sum_{e\in G[B]}\frac{x(e)}{1+\eps}\leqslant \frac{\card{B}}{2\cdot (1+\eps)}\leqslant \frac{\card{B}-1}{2}
    \end{align*}
This proves the observation.
\end{proof}

\begin{observation2}\label{obs:litteflow}
Let $\vec{f}$ be a fractional matching on a multigraph, that puts flow at most $\eps$ between every pair of vertices $u,v\in V$. Then, $\vec{f}$ satisfies odd set constraints of sets of size at most $\nicefrac{1}{\eps}$. Thus, $\vec{f}$ satisfies all odd-set constraints.
\end{observation2}
\begin{proof}
Consider any odd set $B$ such that $\card{B}\leqslant \nicefrac{1}{\eps}$, then,
\begin{align*}
\sum_{e\in G[B]}f(e)\leqslant \sum_{u,v\in B} f((u,v))\leqslant \eps\cdot \frac{\card{B}\cdot (\card{B}-1)}{2}\leqslant \frac{\card{B}-1}{2}.
\end{align*}
From \Cref{obs:smallodd} we can conclude that $\frac{\vec{f}}{1+\eps}$ satisfies all odd set constraints. 
\end{proof}


\section{Reduction To Multigraphs with Heavy Matchings}
\label{sec:redc}
In this section, we prove the following theorem.

\simpletomulti*

\begin{observation2}
Let $G$ be a simple graph, then it contains at most $2^{O(k\cdot \log n)}$ matchings which have $k$ edges.
\end{observation2}
\begin{proof}
This is evident from the fact that there are at most ${n^2 \choose k}$ ways of choosing $k$ edges from among $n^2$ edges in total. 
\end{proof}

Now, in order prove \Cref{lem:simpletomulti}, we give a simple procedure that takes as input a simple graph, and outputs a multigraph on $O(\nicefrac{\textsf{mwm}(G)\cdot W}{\eps})$ vertices that preserves the weight of a fixed matching $M$ of $G$ with probability at least $1-\exp\paren{-O(\nicefrac{\card{V(M)}\cdot \eps^3}{W^3})}$.

\begin{algorithm}[H]
	\algorithmicrequire{ Graph $G$, a sparsification parameter $\tau$, and a parameter $\eps>0$}\\
	\algorithmicensure{ A multigraph $\mathcal{G}(\mathcal{V},\mathcal{E})$ with $\card{\mathcal{V}}=\tau$}
	\caption{\textsc{Vertex-Red-Basic}($G,\tau,\eps$)}
	\begin{algorithmic}[1]
		\State Partition $V$ into $\tau$ bins, $\mathcal{V}\coloneqq (B_1,B_2\cdots, B_{\tau})$, by assigning every vertex to one of the $\tau$ bins uniformly at random. 
		\State For a vertex $u$, let $B(u)$ denote the bin chosen for $u$. For any edge $(u,v)\in E$, if $B(u)\neq B(v)$, then add an edge $e_{u,v}$ between $B(u)$ and $B(v)$.
		\State Return the multigraph $\mathcal{G}(\mathcal{V},\mathcal{E})$.
	\end{algorithmic}
	\label{alg:vertexred}
\end{algorithm}

\begin{lemma2}\label{lem:analysisalg}
Let $G$ be a simple weighted graph with weights in $\set{1,\cdots, W}$, and let $\tau= \frac{8\cdot (1+\delta)\cdot W\cdot \textsf{mwm}(G)}{\eps}$. Let $\mathcal{G}=\textsc{Vertex-Red-Basic}(G,\tau,\eps)$, and let $M$ be any fixed matching of $G$. Then, there exists a matching $\mathcal{M}\subset M$ in $\mathcal{G}$ such that $w(\mathcal{M})\geqslant (1-\eps)\cdot w(M)$ with probability at least $1-2^{-\frac{\card{V(M)}\cdot \eps^3}{32\cdot W^3}}$.
\end{lemma2}
\begin{proof}
Let $\delta = \nicefrac{\varepsilon}{4}$, and let $t= 2\cdot W\cdot \card{V(M)}$. Note that by our definition, $\delta<\nicefrac{1}{2}$. Recall that there are $\tau= \frac{8\cdot (1+\delta)\cdot W\cdot \textsf{mwm}(G)}{\eps}$ bins, we combine these arbitrarily so that there are now $\nicefrac{t}{\delta}$ groups. Call these groups $Z_1,\cdots, Z_{\nicefrac{t}{\delta}}$.  Note that every vertex $v$ lands in $Z_i$ with probability $\nicefrac{\delta}{t}$. We call a group bad if it doesn't contain even one of the $2\cdot \card{V(M)}$ vertices of $M$. Let $X_i$ be the random variable that takes value $1$ if $Z_i$ is bad, and is $0$ otherwise. So, we have,
\begin{align*}
\prob{X_i=1}&=\paren{1-\frac{\delta}{t}}^{2\cdot \card{V(M)}}\\
&\leqslant e^{-\nicefrac{\delta}{W}}\\
&\leqslant 1-\nicefrac{\delta}{W} +\nicefrac{\delta^2}{2\cdot W^2}
\end{align*}
Let $X$ be a random variable denoting the number of bad groups, that is, $X=\sum_{i=1}^{\nicefrac{t}{\delta}} X_i$. Note that we have,
\begin{align*}
\expect{X}\leqslant\nicefrac{t}{\delta}(1-\nicefrac{\delta}{W} +\nicefrac{\delta^2}{2\cdot W^2}). 
\end{align*}
Note that $X$ is a sum of negatively correlated random variables. Thus, using Chernoff bound, we have,
\begin{align*}
\prob{X\geqslant (1+(\nicefrac{\delta}{W})^2)\cdot \nicefrac{t}{\delta}\cdot (1-\nicefrac{\delta}{W} +\nicefrac{\delta^2}{2\cdot W^2})}&\leqslant \exp\paren{-(\nicefrac{\delta}{W})^4\cdot \nicefrac{t}{\delta}\cdot (1-\nicefrac{\delta}{W} +\nicefrac{\delta^2}{2\cdot W^2}) }\\
&=\exp\paren{-\frac{\delta^3\cdot t}{W^4}+\frac{\delta^4\cdot t}{W^5}}\\
&\leqslant \exp\paren{-\frac{\delta^3\cdot t}{2\cdot W^4}}
\end{align*}
Thus, with probability at least $1- \exp\paren{-\frac{\delta^3\cdot t}{2\cdot W^4}}$, we have, $X\leqslant \nicefrac{t}{\delta}\cdot (1-\nicefrac{\delta}{W}+\nicefrac{2\delta^2}{W^2})$. This implies that the number of good groups is at least $2\card{V(M)}-\nicefrac{4\cdot \delta}{W}\cdot \card{V(M)}$. Now, we want to show there is a matching $\mathcal{M}$ among these groups such that $w(\mathcal{M})\geqslant (1-4\cdot \delta)\cdot w(M)$. To see this, in every group, we fix one vertex of $M$, and remove the rest. This way, we deleted at most $\nicefrac{4\cdot \delta}{W}\cdot \card{V(M)}$ edges of $M$, and these can contribute total weight at most $4\cdot \delta \cdot\card{V(M)}\leqslant 4\cdot \delta \cdot w(M)$ to $M$. Thus, we have in the remaining matching, which we call $\mathcal{M}$, $w(\mathcal{M})\geqslant (1-4\cdot \delta)\cdot w(M)$. 
\end {proof}

\begin{lemma2}\label{lem:matchpreserve}
Let \textsc{Vertex-Red()} be an algorithm that does independent runs of \textsc{Vertex-Red-Basic}(). Let $H_1,\cdots, H_{\lambda}$ be multigraphs obtained on independent runs of \textsc{Vertex-Red-Basic}() on input a weighted graph $G$ with weights in $\set{1,\cdots, W}$, $\tau\geqslant \frac{4\cdot \text{mwm}(G)\cdot W}{\eps}$ and $\eps\in(0,\nicefrac{1}{2})$, for $\lambda \geqslant \frac{1000\log n\cdot W^3}{\eps^4\cdot (1-\eps)}$. Then, with probability at least $1-O(\nicefrac{1}{n^2})$, for all matchings $M$ of $G$, there is a matching $M'$ in some $H_i$ such that $M'\subset M$ and $w(M')\geqslant (1-\eps)\cdot w(M)$.  
\end{lemma2}
\begin{proof}
Consider a fixed matching $M$, with $\card{M}=k$.  Then, with probability at least $1-\exp\paren{-\frac{\eps^3\cdot k}{32\cdot W^3}}$, a fixed $H_i$ contains $M'\subset M$, with $w(M')\geqslant (1-\eps)\cdot w(M)$.  Since each of the $H_j$ are independent runs of \textsc{Vertex-Red-Basic}(), with probability at least  $1-\exp\paren{-\frac{k\cdot \log n}{\eps}}$, some $H_i$ contains a matching $M'$ such that $M'\subset M$, and $w(M')\geqslant (1-\eps)\cdot w(M)$. Taking a union bound over all possible matchings with $k$ edges, and then over all possible values of $k$, we have the claim we need. 
\end{proof}
We now show the following. 
\begin{lemma2}\label{lem:matchpreserve2}
Let $G$ be a simple weighted graph, with weight function $w$ on the edges, and weights in $\set{1,2,\cdots, W}$. Let $H_1,\cdots, H_{\lambda}$ be output of independent runs of \textsc{Vertex-Red-Basic}($G,\lambda, \eps$), where $\lambda \geqslant \frac{1600\cdot\log (n\cdot W)}{\eps^4\cdot (1-\eps)}$. Consider an adversary that deletes edges from $G$. Suppose we simulate these deletions on each of the $H_i$. That is, if $(u,v)$ is deleted from $G$, then $e_{u,v}$ (if present) is deleted from each of the $H_i$. Let $\textsf{mwm}_{i}^t$ denote the weight of $\textsf{mwm}(H_i)$ at time $t$, and let $\textsf{mwm}^t$ denote the weight of $\textsf{mwm}(G)$ at time $t$. Suppose $\textsf{mwm}^t\geqslant (1-\eps)\cdot \textsf{mwm}^0$, then $\textsf{mwm}_i^t\geqslant (1-2\eps)\cdot \textsf{mwm}^0$ for some $i\in [\lambda]$.
\end{lemma2}
\begin{proof}
Let $M$ be a matching satisfying the assumption of the lemma. That is, let $M$ be a matching at time $t$ in $G$ such that $w(M) \geqslant (1-\eps)\cdot \textsf{mwm}^0$. Since the adversary is only performing deletions, thus we have, initially, before all deletions are performed, $w(M)\geqslant (1-\eps)\cdot \textsf{mwm}^0$. Thus, by \Cref{lem:matchpreserve}, we know that there is an $H_i$ containing $M'\subset M$ with the property that $w(M')\geqslant (1-\eps)\cdot w(M)$. Since all edges of $M$ survive at time $t$, this is true for $M'$ as well. Consequently, we have that at time $t$, $w(M')\geqslant (1-\eps)^2\cdot \textsf{mwm}^0$. This proves our claim. 
\end{proof}

\begin{proof}[Proof of \Cref{lem:simpletomulti}]
 \Cref{lem:simpletomulti}\ref{item:parta} follows from the properties of \textsc{Vertex-Red}(),  \Cref{lem:analysisalg}, and \Cref{lem:matchpreserve}. On the other hand, \Cref{lem:simpletomulti}\ref{item:partb} follows from the contrapositive of \Cref{lem:matchpreserve2}. 
\end{proof}

\section{Properties of Static-Weighted-Match}

In this section, we show \Cref{lem:mwm}, which we first restate.

\swm*

We first note that the properties \ref{item:mwma}-\ref{item:mwmb} and \ref{item:mwmd}-\ref{item:mwmf} are evident in Property 3.1 and the proof of Lemma 2.3 of \cite{DP14}. We now show that \ref{item:mwmc} also holds.

\begin{proof}[Proof of \Cref{lem:mwm}]
Let $\textsc{Basic-Static-Weighted-Match}()$ be the algorithm satisfying of \cite{DP14} satisfying \ref{item:mwma}-\ref{item:mwmb} and \ref{item:mwmd}-\ref{item:mwmf}. We now show how to obtain \textsc{Static-Weighted-Match}() from \textsc{Basic-Static-Weighted-Match}(). We run \textsc{Basic-Static-Weighted-Match}() with input $G$ and $\delta=\nicefrac{\eps}{3}$. Thus, we get a matching $M$ with $w(M)\geqslant (1-\nicefrac{\eps}{3})\cdot \textsf{mwm}(G)$. Moreover, the duals $\vec{y}$ and $\vec{z}$ output by the algorithm have $f(y,r)\leqslant (1+\nicefrac{\eps}{3})\cdot \textsf{mwm}(G)$. We consider the following procedure: for every $B\in \Omega$ with $z(B)>0$ and $\card{B}\geqslant \frac{3}{\eps}+1$, we increase $y(v)$ by $\frac{z(B)}{2}$ for every $v\in B$, and we decrease $z(B)$ to $0$. Thus, each $y(v)$ is still a multiple of $\eps$ and each $z(B)$ is still a multiple of $\eps$. We conclude that \ref{item:mwmd} still holds. Moreover, this transformation keeps the value of the dual constraint for every edge the same. Thus, we still satisfy \ref{item:mwme}. We update $\Omega$ by removing $B$ from it. The set $\Omega$ still remains laminar. Thus, \ref{item:mwmb} is still satisfied. We first observe that the time taken to do this is linear in the sum of the sizes of the odd sets $B$ with $z(B)>0$. Note that since $yz(V)$ sums to at most $(1+\nicefrac{\eps}{3})\textsf{mwm}(G)$, and non-zero $z(B)$ have value at least $\nicefrac{\eps}{3}$, this implies that the sums of the sizes of the odd sets is at most $\nicefrac{3}{\eps}\cdot (1+\nicefrac{\eps}{3})\cdot \textsf{mwm}(G)$. Thus, the time taken for the procedure is $O(\nicefrac{m}{\eps})$. Next, we observe that the value of $yz(V)$ changes by at most $\sum_{B:\card{B}\geqslant \nicefrac{3}{\eps}+1}\frac{z(B)}{2}$. This value is at most $\nicefrac{\eps}{3}\cdot (1+\nicefrac{\eps}{3})\cdot \textsf{mwm}(G)$, as shown by the following calculation.
	
	\begin{equation*}
	\begin{split}
		\paren{\frac{3}{\eps}}\cdot \sum_{B:\card{B}\geqslant \nicefrac{3}{\eps}+1}\frac{z(B)}{2}\leqslant \sum_{B:\card{B}\geqslant \nicefrac{3}{\eps}+1}\paren{\frac{\card{B}-1}{2}}\cdot z(B)&\leqslant (1+\nicefrac{\eps}{3})\cdot \textsf{mwm}(G)
	\end{split}
	\end{equation*}
	
	Thus, the new $yz(V)$ has value at most $(1+\nicefrac{\eps}{3})^2\cdot \textsf{mwm}(G)$. This implies that \ref{item:mwmf} holds. Finally, since every blossom $B$ of size at least $\nicefrac{3}{\eps}+1$ has $z(B)=0$, thus, \ref{item:c} holds. 
\end{proof}

\section{Matchings in Decremental Graphs with \textsf{mwm}($G$)$=O(\log n)$}
\label{section:smallmatching}

As mentioned in the introduction, the underlying assumption in our paper is that the decremental graph $G$ has $\textsf{mwm}(G)=\Omega(\log n)$ as mentioned in \Cref{assumption:largematching}. We now give a simple decremental algorithm for the case where the graph $G$ has $\textsf{mwm}(G)=O(\log n)$. We start by recalling \Cref{lem:guptapeng}, which allows us to focus on designing an algorithm for the case where the weights are integers in $\set{1,\cdots, W}$. Using this assumption, we have the following observation.

\begin{observation2}
If $G$ is an integer-weighted graph with weight function $w$ on the edges, with weights in $\set{1,2,\cdots,W}$. Then, $\textsf{mwm}(G)\geqslant \textsf{mcm}(G)$. Consequently, if $\textsf{mwm}(G)=O(\log n)$, then there is a vertex cover of $G$ of size $O(\log n)$. 
\end{observation2}

Given a graph $G$, we give a weighted matching sparsifier for $G$. This sparsifier is mentioned in \cite{GP13} for unweighted graphs, but the proof extends to weighted graphs as well, and we state it here for completeness.

\begin{definition2}
Given a graph $G$, and a vertex cover $S$, the core graph $H$ consists of the following edges:
\begin{enumerate}
\item All edges in $G[S]$.
\item For each vertex $v\in S$, $|S|+1$ heaviest edges from $v$ to $V\setminus S$. 
\end{enumerate}
\end{definition2}

\begin{lemma2}\cite{GP13}
Given a graph $G$, and a core graph $H$ of $G$, we have $\textsf{mwm}(H)=\textsf{mwm}(G)$.
\end{lemma2}
\begin{proof}
Suppose $\textsf{mwm}(G)>\textsf{mwm}(H)$. Let $M^*$ denote the optimal matching in $G$ that uses the most number of edges from $H$. There is an edge $e\in M^*\setminus H$. Observe that since $S$ is a vertex cover, and we have $G[S]\subseteq H$, this implies that $e$ has one end point in $u\in S$ and the other in $v\in V\setminus S$. Note that $e$ wasn't included into $H$ since there $u$ has at least $\card{S}+1$ edges incident on it in $H$, and for each of these edges $e'$, we have $w(e')\geqslant w(e)$. Next, observe that one of these neighbours of $u$ in $V\setminus S$ must be unmatched (call this $w$). This is because $S$ is a vertex cover, and there are no edges among vertices in $V\setminus S$, and consequently, at most $|S|$ vertices in $V\setminus S$ can be matched. Substituting $(u,v)$ with $(u,w)$ in $M^*$, we get a matching $M'$ in $G$ which uses even more edges from $H$ and $w(M')\geqslant w(M^*)$. This is a contradiction. 
\end{proof}

\begin{lemma2}
Consider a decremental graph $G$ with $\textsf{mwm}(G)=O(\log n)$. Then, there is a decremental algorithm for maintaining a $(1-\varepsilon)-\textsf{mwm}(G)$ in $O(\log^2n)$ update time.
\end{lemma2}
\begin{proof}
Given a decremental graph $G$ with $\textsf{mwm}(G)=O(\log n)$. We can in $O(m)$ time preprocess the graph to sort the edges of every vertex and store these as $n$ doubly linked lists. We can in $O(m)$ time compute a maximal matching of $G$, and using this in compute a vertex cover $S$ of $G$ with $|S|=O(\log n)$. Recall, the $G$ is decremental, therefore, $S$ always remains a vertex cover of $G$. Thus, we never need to recompute $S$. Next, we a compute the core graph $H$ in $O(\log^2n)$ time. Using a standard approximate weighted matching algorithm, we can compute a $(1-\varepsilon)$ approximate matching $M$ of $H$ in time $O(\log^2n)$ (see \Cref{lem:staticmatch}). As we process deletions, if an edge is deleted from $G\setminus H$, we just remove that edge from the list and do nothing. If an edge is deleted from $H$, then in addition to doing the above-mentioned step, we also delete the edge from $H$. If this was an edge from a vertex $u\in S$ to a vertex $v\in V\setminus S$, then we add to $H$, the next highest weighted edge from $u$ to $V\setminus S$ and update the doubly-linked list in $O(1)$ time. Finally, if this edge was also in $M$, we just recompute a new matching of $H$ from scratch in time $O(\log^2n)$ since $\card{E(H)}=O(\log^2n)$.
\end{proof}

\section{Rounding Weighted Fractional Matchings}
\label{sec:rounding}
In this section, we will prove \Cref{lem:sparsification}. More concretely, we give state the procedure \textsc{Sparsification}() that takes as input a fractional matching $\vec{x}$ of a simple graph $G$, and outputs a graph $H$, of size $\Tilde{O}(\mu(G))$. If $x(e)\leqslant \eps^6$ for all $e\in E(G)$, then, $H$ contains an integral matching of weight at least $(1-10\eps)\cdot \sum_{e\in E}w(e)\cdot x(e)$ in its support. The proof of this theorem is implicit in the work of \cite{Wajc2020}, but we describe it here for completeness. We first state the algorithm, and then describe some of its properties. Note that even though we work over a multigraph, and compute a fractional matching $\vec{x}$, we feed $\vec{x}^C$ into the rounding algorithm. Fractional matching $\vec{x}^C$ has the property that between every pair of vertices $u,v\in V$, there is at most one edge of a given weight. Thus, the total number of edges is at most $n^2\cdot W$. 


The algorithm uses a dynamic edge coloring algorithm as a subroutine. The update time of the subroutine is $O(\log n)$ in the worst case. The sparsification algorithm takes as input a fractional matching $\vec{x}$ and a parameter $\eps>0$. Then, it classifies $E(G)$ into classes as follows: $E_i=\set{e\mid x(e)\in [(1+\eps)^{-i}, (1+\eps)^{-i+1})}$. Additionally, we can limit the number of edge classes into $O(\log nW)$ many using the following argument.

\begin{observation2}
    We can assume that $\frac{\min_{e\in E} x(e)}{\max_{e\in E}x(e)}\geqslant \frac{\eps^2}{n^2\cdot W}$. 
\end{observation2}
\begin{proof}
Let $F=\set{e\in E\mid x(e)\leqslant \frac{\eps^2}{n^2\cdot W}\cdot \max_{e'\in E}x(e')}$. Note that $\sum_{e\in F}w(e)\cdot x(e)\leqslant \eps^2\cdot \max_{e'\in e}x(e')$. Thus, we can throw away these edges without hurting the weight of the fractional matching by more than an $1-\eps$ factor.
\end{proof}
From the above observation, we can conclude that the total number of edge classes are at most $\log(\nicefrac{n\cdot W^2}{\eps})$ many.
\begin{algorithm}
	\caption{\textsc{Sparsification}($\vec{x},\eps$)}
	\begin{algorithmic}[1]
		\State $d\leftarrow \frac{4\cdot \log\paren{\nicefrac{2}{\eps}}}{\eps^2}$
		\For{$i\in \set{1,2,\cdots, 2\log_{1+\eps}(\nicefrac{n\cdot W}{\eps})}$}
		\State Compute a $3\lceil (1+\eps)^i\rceil$-edge colouring $\Phi_i$ of $E_i$.
		\State Let $S_i$ be a sample of $3\cdot \min\set{\lceil d\rceil,\lceil (1+\eps)^i\rceil}$ colours without replacement in $\Phi_i$. \label{line:blah}
		\State Return $K=(V,\cup_{i}\cup_{M\in S_i}M)$
		\EndFor
	\end{algorithmic}
\end{algorithm}

We state the guarantees of the edge coloring subroutine. 

\begin{observation2}
	The size of $K$ output by \textsc{Sparsification}($\vec{x},\eps$) is,
	\begin{align*}
	\card{E(H)}=O\paren{\frac{\log (\nicefrac{n\cdot W}{\eps})}{\eps}\cdot d\cdot \mu(\textup{supp}(H))}
	\end{align*}
\end{observation2}

\begin{lemma2}\label{lem:edgecolour}\cite{BDHN18}
There is a deterministic dynamic algorithm that maintains a $2\Delta-1$ edge coloring of a graph $G$ in $O(\log n)$ worst case update time, where $\Delta$ is the maximum degree of the graph $G$. 
\end{lemma2}

\begin{observation2}\label{obs:sampling}
Let $\eps\in (0,\nicefrac{1}{2})$ and suppose the input to \textsc{Sparsification}($\vec{x},\eps$) is a matching $\vec{x}$ with $x(e)\leqslant \eps^6$, then, $\card{S_i}=3\cdot \lceil d\rceil$. 
\end{observation2}

To show that $K$ contains a matching of weight at least $(1-\eps)\cdot \sum_{e\in E(G)}w(e)\cdot x(e)$ in its support, we will show a random fractional matching $\vec{y}$ in $K$ that sends flow at most $\eps$ through each of its edges. Moreover, $\expect{\sum_{e\in E}w(e)\cdot y(e)}\geqslant (1-6\eps)\cdot \sum_{e\in E(G)}w(e)\cdot x(e)$. This will show the existence of a large matching that sends flow at most $\eps$ through each of its edges. To show this, we give some properties of the algorithm.

\begin{lemma2}\label{lem:probbounds}
Let $\eps\in (0,\nicefrac{1}{2})$ and let $\vec{x}$ be the input to \textsc{Sparsification}($\vec{x},\eps$) such that $x(e)\leqslant \eps^6$ for all $e\in E(G)$. Then, for every edge $e$, $\prob{e\in K}\in \bracket{\nicefrac{x(e)\cdot d}{(1+\eps)^2},x(e)\cdot d\cdot (1+\eps)}$.
\end{lemma2}
\begin{proof}
Since $\vec{x}$ has the property that for every $e\in E$, $x(e)\leqslant \eps^6$, this implies that we sample $3\cdot \lceil d\rceil$ from $\Phi_i$ colors for each $i$ (from Observation \ref{obs:sampling}). Thus, if we consider an edge $e\in E_i$, then, we have,
\begin{align*}
\prob{e\in K}=\frac{\lceil d\rceil}{\lceil(1+\eps)^i\rceil}\leqslant \frac{d\cdot (1+\eps)}{(1+\eps)^i}\leqslant d\cdot (1+\eps)\cdot x(e)
\end{align*}
The last inequality follows from the fact that $x(e)\in \bracket{(1+\eps)^{-i},(1+\eps)^{-i+1}}$. Moreover, we have,
\begin{align*}
	\frac{\lceil d\rceil}{\lceil (1+\eps)^i\rceil}\geqslant \frac{d}{(1+\eps)^i+1}
	\geqslant \frac{d}{(1+\eps)^{i+1}}
	\geqslant \frac{d\cdot x(e)}{(1+\eps)^2}
	\end{align*}
The first inequality follows from the fact that $\lceil d\rceil \geqslant d$, and $\lceil (1+\eps)^i\rceil\leqslant (1+\eps)^i+1$. The second inequality follows from the following reasoning.
\begin{align*}
(1+\eps)^{-i}\leqslant x(e)\leqslant \eps^6\leqslant \nicefrac{1}{d}\leqslant \eps
\end{align*}
Thus, $\nicefrac{1}{\eps}\leqslant (1+\eps)^i$, and $1\leqslant \eps\cdot (1+\eps)^i$, so the second inequality follows. The last inequality follows from the fact that $x(e)\in \bracket{(1+\eps)^{-i},(1+\eps)^{-i+1}}$.
\end{proof}

For an edge $e\in E(G)$, we define $X_e$ to be an indicator random variable that takes value $1$ if $e\in K$, and $0$ otherwise. 

\begin{observation2}\label{obs:negatcorr}
	For a vertex $v$, the variables $\set{X_e\mid e\text{ incident on }v}$ are negatively associated. 
\end{observation2}
\begin{proof}
	If the edges $e$ and $e'$ incident on $v$ belong in $E_i$ and $E_j$ where $i\neq j$, then $X_e$ and $X_{e'}$ are independent. On the other hand, if they belong in the same $E_i$, then recall we did an edge colouring on $E_i$, so $e'$ and $e$ got different colours. So, if the colour corresponding to $e$ is picked into $K$, then this reduces the probability of the colour corresponding to $e'$ being picked into $K$.
\end{proof}

\begin{lemma2}\label{cor:cond}
	From Observation \ref{obs:negatcorr} and \Cref{lem:probbounds}, we can conclude that for edges $e$ and $e'$ incident on $v$, 
	\begin{align*}
	\prob{X_e=1\mid X_{e'}=1}\leqslant \prob{X_e=1}\leqslant x(e)\cdot d\cdot (1+\eps)
	\end{align*}
\end{lemma2}

\begin{lemma2}
	For any vertex $v$ and edge $e'$ incident on $v$, the variables $\set{[X_e\mid X_{e'}]\mid e\text{ incident on }v}$ are negatively associated.
\end{lemma2}

Following theorem directly implies \Cref{lem:sparsification}. 

\begin{theorem2}\label{thm:largeintegral}
Let $K$ be the subgraph of $G$ output by \textsc{Sparsification}(), when run on the matching $\vec{x}$ with parameters $\eps\in (0,\nicefrac{1}{2})$. If $x(e)\leqslant \eps^6$ for all $e\in E$, then $K$ supports a fractional matching $\vec{y}$ such that $y(e)\leqslant \eps$, and $\sum_{e\in E}w(e)\cdot y(e)\geqslant (1-6\eps)\cdot \sum_{e\in E}w(e)\cdot x(e)$.
\end{theorem2}
\begin{proof}
	To show the theorem, we will describe a process that simulates \textsc{Sparsification}($\vec{x},\eps$), and outputs a random matching $\vec{y}$ such that $y(e)\leqslant \eps$ for every $e\in E$. Additionally, $\expect{\sum_{e\in E}w(e)\cdot y(e)}\geqslant (1-6\eps)\sum_{e\in E}w(e)\cdot x(e)$. As an intermediate step, let $z(e)=\frac{(1-4\eps)}{d}\cdot X_e$. So, we have,
	\begin{align*}
	\expect{z(e)}&\geqslant \expect{z(e)\mid X_e=1}\cdot \prob{X_e=1}\geqslant \frac{1-4\eps}{d}\cdot \frac{x(e)\cdot d}{(1+\eps)^2}\geqslant (1-6\eps)\cdot x(e)
	\end{align*}
	
The second to last inequality follows from \Cref{lem:probbounds}. Next define $\vec{y}$ as follows:
\begin{align*}
y(e)=\begin{cases}
0 \text{ if for one of the endpoints }v\text{ of }e,\  \sum_{e'\in v}z(e')>1\\
z(e) \text{ otherwise.}
\end{cases}
\end{align*}
Essentially, our procedure is creating matching $\vec{z}$ as follows: it assigns value $\frac{1-4\eps}{d}$ to $z(e)$ if $e$ was included in $K$, and $0$ otherwise. The value $y(e)$ is the same $z(e)$ in all cases except when $e$ is picked into $K$, but at one of the endpoints $v$, $\sum_{e'\ni v}z(e')>1$ , that is, the fractional matching constraint is violated at $v$. We show that it is unlikely that an edge (when picked) has one of its endpoints violated. Let $e'$ be an edge incident on $v$, we to bound the probability that $z(e')\neq y(e')$. 
\begin{align*}
\expect{\sum_{e\in v}z(e)\mid X_{e'}=1}&\leqslant \frac{1-4\eps}{d}+\expect{\sum_{e'\neq e, e\in v}z(e)\mid X_{e'}}\\
&\leqslant \eps+\sum_{e\in v}x(e)\cdot (1+\eps)\cdot d\cdot\paren{\frac{1-4\eps}{d}}\\
&\text{(Since $\nicefrac{1}{d}\leqslant \eps$ and from \Cref{cor:cond})}\\
&\leqslant (1-\eps)
\end{align*}
Note that at any end point of $e'$, conditioned on $e'$ being sampled, the expected sum of $z(e)$'s at that endpoint is upper bounded by $(1-\eps)$. Now, $z(e)$ is assigned value $0$ only if the sum of the $z(e)$'s deviates from the expected value by $\eps$. To see this, we want to compute $\var{[z(e)\mid X_{e'}]}$, where $e$ and $e'$ share an end point. Note that $[z(e)\mid X_{e'}]$ takes value $\frac{1-4\eps}{d}$ with probability $\prob{X_e\mid X_{e'}}$, otherwise it takes value $0$.
\begin{align*}
\var{\bracket{z(e)\mid X_{e'}}}&\leqslant \expect{\bracket{z(e)\mid X_{e'}}^2}\\
&\leqslant\paren{\frac{1-4\eps}{d}}^2\cdot \prob{X_e\mid X_{e'}}\\
&\leqslant \paren{\frac{1-4\eps}{d}}^2\cdot x(e)\cdot d\cdot (1+\eps)\\
&\text{(From \Cref{cor:cond})}\\
&\leqslant \frac{x(e)}{d}
\end{align*}
This implies that $\sum_{e\in v}\var{\bracket{z(e)\mid X_{e'}}}\leqslant \frac{1}{d}$. So, we want to compute the probability that the sum of the random variables $\set{\bracket{z(e)\mid X_{e'}}}$ deviates from the expected value by $\eps$. Applying \Cref{lem:concentrate}, we have,
\begin{align*}
\prob{\sum_{e\in v}\bracket{z(e)\mid X_{e'}}\geqslant \expect{\sum_{e\in v}\bracket{z(e)\mid X_{e'}}}+\eps}&\leqslant \exp\paren{\frac{-\eps^2}{2\paren{\frac{1}{d}+\frac{\eps}{3\cdot d}}}}\\
&\paren{\text{Since }z(e)\in \bracket{0,\nicefrac{1-4\eps}{d}}}\\
&\leqslant \exp\paren{-\eps^2\cdot 0.25\cdot d}\\
&\text{(Since }\eps\in (0,1))\\
&\leqslant \frac{\eps}{2}\\
&\text{(Since }d=\frac{4\cdot \log(\nicefrac{2}{\eps})}{\eps^2})
\end{align*}
Taking union bound over both endpoints, we know that $\prob{y(e)=z(e)\mid X_e=1}\geqslant (1-\eps)$. Thus, we have:
\begin{align*}
\expect{y(e)}&=\paren{\frac{1-4\eps}{d}}\cdot \prob{y(e)=z(e)}\\
					&=\paren{\frac{1-4\eps}{d}}\cdot \prob{y(e)=z(e)\mid X_e=1}\prob{X_e=1}\\
					&\geqslant \paren{\frac{1-4\eps}{d}}\cdot (1-\eps)\cdot \frac{x(e)\cdot d}{(1+\eps)^2}\\
					&\geqslant (1-7\eps)\cdot x(e)
\end{align*}
Thus, $\expect{\sum_{e\in E}w(e)\cdot y(e)}\geqslant (1-7\eps)\cdot \sum_{e\in E}w(e)\cdot x(e)$. Moreover, for all $e\in E$, $y(e)\leqslant \eps^6$. This proves our claim. 
\end{proof}

We now restate the lemma, and then show its proof. 

\lemuno*
\begin{proof}
Note that \ref{lem:sparseb} is implied by \Cref{thm:largeintegral}, and the fact that any fractional matching $\vec{y}$ with $y(e)\leqslant \eps$ for all $e\in E$, satisfies odd set constraints for all odd sets of size at most $\nicefrac{1}{\eps}$. To see \ref{lem:sparsea}, note that deleting $e$ from supp($\vec{x}$) corresponds to just deleting $e$ from $G_i$, and since our edge coloring algorithm is able to handle edge insertions and deletions in $\Tilde{O}_{\eps}(1)$ time, such updates can be handled in $\Tilde{O}_{\eps}(1)$ update time (see \Cref{lem:edgecolour}). Finally, if an update reduces $x(e)$ for some $e\in E$, then this corresponds to deleting $e$ from some $G_i$ and adding it to $G_j$ for some $j<i$. Thus, the edge colouring algorithms running on $G_i$ and $G_j$ have to handle an edge insertion and deletion respectively, and this can be done in $\Tilde{O}_{\eps}(1)$ time (see \Cref{lem:edgecolour}). 
\end{proof}

\newpage
\bibliographystyle{alpha}
\bibliography{references.bib}
\end{document}

%% file: introduction.tex
\section{Introduction}

In dynamic algorithms, the main goal is to maintain a key property of the graph while an adversary makes changes to the edges of the graph. An algorithm is called incremental if it only handles edge insertions, and decremental if it handles only edge deletions and fully dynamic if it handles both edge insertions and deletions. \\ \\
The problem of maintaining a $(1-\eps)$-approximation to the maximum matching in a dynamic graph is a well-studied one. In the fully dynamic setting, the best known update time is $O(m^{0.5-\Omega_{\eps}(1)})$ (see \cite{BKS23}), and the conditional lower bounds proved in the \cite{HKNS15,KPP16} suggest that this is a hard barrier to break through. Consequently, several relaxations of this problem have been studied. For example, one line of research has shown that we can get considerably faster update times if we settle for large approximation factors. Another research direction has been to consider more relaxed incremental or decremental models. In the incremental setting, there have been a series of upper and lower bound results (see \cite{Dahlgaard16} and \cite{Gupta2014}), culminating in the results of \cite{GLSSS19} (who give a $O(m\cdot (\nicefrac{1}{\eps})^{O(\nicefrac{1}{\eps})})$ total update time algorithm for $(1-\eps)$-approximate matching in \textbf{general graphs}) and of \cite{BlikstadK23} (who give a $O(m\cdot \poly(1/\eps))$ total update time algorithm for $(1-\eps)$-approximate matching in \textbf{bipartite graphs}). In contrast, for decremental matching \cite{BPT20,JambulapatiJST22} gave an $O(m\cdot\poly(\log n,\nicefrac{1}{\eps}))$ total update time algorithm for $(1-\eps)$-approximate matching in bipartite graphs, and \cite{AssadiBD22} extended this algorithm to the case of general graphs by giving an algorithm for $(1-\eps)$-approximate decremental matching with total time $O_{\eps}(m\cdot \poly(\log n))$, thus completing the picture for approximate cardinality matching in the partially dynamic setting. \\ \\
For the weighted case (in both fully and partially dynamic settings) the picture is less clear. \cite{BernsteinDL21} showed that for the bipartite case, it is possible to reduce weighted matching to unweighted matching with an overhead of $(1-\eps)$ in the approximation factor, and at a loss of a factor of $\log W$ and $(\nicefrac{1}{\eps})^{O(\nicefrac{1}{\eps})}$ in the update time. Recently, \cite{BhattacharyaKS23b} showed that this exponential-in-$1/\eps$ dependence was avoidable for bipartite graphs in the partially dynamic setting. \\ \\
In contrast, for non-bipartite graphs, we don’t have such a general reduction. Thus, while there has been significant recent progress on unweighted matching, the state-of-art for weighted matching lags far behind. Our main contribution in this paper, is to close this gap between weighted and unweighted matching for general graphs in the decremental dynamic setting:



\begin{restatable}{theorem2}{mainthm}\label{thm:main}
Let $G$ be a weighted graph with weight function $w$ on the edges of the graph and let $\eps\in (0,\nicefrac{1}{2})$. Then, there is a decremental algorithm with total update time $O_{\varepsilon}(m\cdot\poly(\log n)\cdot\log R)$ (amortized $\Tilde{O}_{\varepsilon}(1)$) that maintains an integral matching $M$ with $w(M)\geqslant (1-\varepsilon)\cdot \textsf{mwm}(G)$, with high probability, where $G$ refers to the current version of the graph. The algorithm is randomized but works against an adaptive adversary. The dependence on $\eps$ is $2^{{\paren{\nicefrac{1}{\varepsilon}}}^{O\paren{\nicefrac{1}{\varepsilon}}}}$.
\end{restatable}

In comparison, the decremental algorithm for unweighted $(1-\eps)$ approximate matching, is also randomized but works against an adaptive adversary, and takes the same total time albeit without an additional loss of a $\log R$ factor and a dependence of $(\nicefrac{1}{\eps})^{O(\nicefrac{1}{\eps})}$ on $\eps$.

\subsection{Our Techniques}

Several unweighted dynamic algorithms maintain a matching sparsifier, which is a sparse subgraph of $G$ guaranteed to contain a large matching. As mentioned before, unlike in the bipartite case, there is no reduction from weighted to unweighted matchings in general graphs. So, the existing algorithms for weighted matchings in dynamic setting proceed by looking at weighted analogues of unweighted matching sparsifiers. For example, \cite{BernsteinDL21} consider weighted versions of the EDCS and kernels which were first discovered in the context of cardinality matching problem (see \cite{BS15,BS16,BHI15}). We use a similar approach as in \cite{BernsteinDL21} and port the structural results of \cite{AssadiBD22} for unweighted graphs to the setting of arbitrary weighted graphs.

\subsubsection{Structural Results for Decremental Matching}

The congestion balancing framework of \cite{BPT20} maintains a fractional matching that is robust to deletions. In order to ensure that the matching is robust and “spread out”, they impose a capacity function $\kappa$ on the edges of the graph. They repeatedly invoke a subroutine $\text{M-or-E*}()$ that does one of the following (here $\mu(G)$ refers to the size of maximum cardinality matching):

\begin{enumerate}
\item Either output a fractional matching $\vec{x}$, such that $\sum_{e\in E}x(e)\geqslant (1-\eps)\cdot \mu(G)$ or,
\item Output a set $E^*$ of bottleneck edges along which we can increase capacities: \label{prop:congestb}
\begin{enumerate}[label=(\roman*)]
\item $\kappa(E^*)$ is $O(\mu\log n)$\label{prop:congesti}
\item Every large matching in $G$ has $\eps\cdot\mu(G)$ edges in $E^*$. \label{prop:congestii}
\end{enumerate}
\end{enumerate} 
The authors of \cite{BPT20} use this subroutine to get an algorithm for decremental matching as follows: either \textsc{M-or-E*}() outputs a fractional matching $\vec{f}$ such that $\sum_{e\in E}\geqslant (1-\eps)\cdot \mu(G)$, in which case we can use a black-box of \cite{Wajc2020} to round this $\vec{f}$, or if there is no $\vec{f}$ satisfying this property, then it outputs a set of edges $E^*$ along which we can increase capacity. Note that, \ref{prop:congestb}\ref{prop:congesti} ensures that the total capacity increase is small, and \ref{prop:congestb}\ref{prop:congestii} ensures that we only increase capacity along edges necessary to form a large matching. Thus, they are able to show that the capacities remain small on average. \\ \\
Predictably in general graphs, implementing such a subroutine is challenging because of the blossom constraints which impose the following roadblocks:
\begin{enumerate}
\item An arbitrary fractional matching in a general graph can suffer from an integrality gap, which is not the case for bipartite graphs. Consequently, we need to be able to check if a given general graph has a large fractional matching obeying capacities and the blossom constraints. \label{item:issuea}
\item Once we know there is a large fractional matching obeying the above requirements, we need to efficiently find such a fractional matching. In bipartite graphs, this can be done use approximate max flows. However, for general graphs this correspondence between fractional matchings and flows is lost. \label{item:issueb}
\item Finally, in the case of bipartite graphs, if the fractional matching is small, the bottleneck edges $E^*$ along which we have to increase capacity correspond to a min-cut in the graph, however for the case of general graphs, this characterization is not clear.\label{item:issuec}
\end{enumerate}

The main contribution of \cite{AssadiBD22} is to address these challenges through three structural lemmas, we state these informally, and then state our weighted versions of these lemmas.

\begin{lemma2}\label{lem:detect}
Let $G$ be a graph and let $\mu(G)$ denotes the size of the optimal matching in $G$. Suppose $\kappa$ is the capacity function on the edges of the graph, and let $\mu(G,\kappa)$ denote the optimal fractional matching of $G$, obeying capacity function $\kappa$, and all odd set constraints. Sample every edge $e$ with $p(e) \propto \kappa(e)$ to create uncapacitated graph $G_s$. Then, $\mu(G_s)\geqslant \mu(G,\kappa)-\eps\cdot n$.
\end{lemma2}

The above lemma gives us a way to determine whether a given capacitated graph has a large fractional matching obeying blossom constraints by reducing this problem to the problem of finding an integral matching in a general graph. For the latter problem, there are efficient algorithms known. This takes care of Issue \ref{item:issuea}

\begin{lemma2}\label{lem:find}
Let $G$ be a graph with capacity function $\kappa$ on the edges of the graph. Let $G_s$ be the graph obtained by sampling edges independently with probability $p(e)\propto \kappa (e)$. Consider a matching $M$ of $G_s$ such that $\card{M}\geqslant (1-\eps)\cdot \mu(G_s)$. Let $M_H\subseteq M$ be the set of edges of high capacity. Let $V_H=V(M_H)$, Let $V_L=V\setminus V_H$. Compute a fractional matching $\vec{f}$ on the low capacity edges of $G[V_L]$, such that $\vec{f}$ obeys $\kappa$. Then, $\card{M_H}+\sum_{e\in E}f(e)\geqslant \mu(G_s)-\eps\cdot n$.
\end{lemma2}

Thus, by using integral matching algorithms, and approximate maximum flow, the above lemma helps us deal with Issue \ref{item:issueb}.

\begin{lemma2}\label{lem:increase}
The set of bottleneck edges of $E^*$ can be identified by considering the edges of $G$ that are not covered by the duals of the matching program of $G_s$. 
\end{lemma2}

The approximate matching algorithms given in the literature, also give us near feasible dual solution (see \cite{DP09}), thus, the above lemma, in combination with these algorithms gives us a way to deal with Issue \ref{item:issuec}.

\subsection{Weighted Analogues of Structural Results for Decremental Matching}

Similar to the algorithm \textsc{M-or-E*}() in the unweighted decremental setting, we design an algorithm \textsc{WeightedM-or-E*}() that runs in time $O(m)$, and does one of the following two things.

\begin{enumerate}
\item Either output a fractional matching $\vec{x}$, such that $\sum_{e\in E}w(e)\cdot x(e)\geqslant (1-\eps)\cdot \textsf{mwm}(G)$.
\item Output a set E* of bottleneck edges along which we can increase capacities:
\begin{enumerate}[label=(\roman*)]
\item $\sum_{e\in E^*}\kappa(e)\cdot w(e)=O(\textsf{mwm}(G)\cdot\log n)$
\item Every large matching in $G$ has $\eps\cdot \textsf{mwm}(G)$ weight in $E^*$.
\end{enumerate}
\end{enumerate}

The informal statements of the weighted versions of the unweighted lemmas are almost exactly the same:

\begin{lemma2}
Let $G$ be a graph and let $\textsf{mwm}(G)$ denotes the weight of the maximum weight matching in $G$. Suppose $\kappa$ is the capacity function on the edges of the graph, and let $\textsf{mwm}(G,\kappa)$ denote the maximum fractional matching of $G$, obeying capacity function $\kappa$, and all odd set constraints. Sample every edge $e$ with $p(e) \propto \kappa(e)$ to create uncapacitated graph $G_s$. Then, $\textsf{mwm}(G_s)\geqslant \textsf{mwm}(G,\kappa)-\eps\cdot n$.
\end{lemma2}

The proof of \Cref{lem:detect}, relies on the Tutte-Berge Theorem (\cite{Schrijver2003CombinatorialOP}), and the main impediment to proving the above lemma is that there is no weighted analogue of this lemma. Thus, we have to rely on other structural probabilistic tools to come up with a proof.

\begin{lemma2}
Consider a matching M of $G_s$ such that $w(M)\geqslant (1-\eps)\cdot \textsf{mwm}(G_s)$. Let $M_H$ be the set of edges of high capacity. Let $V_H=V(M_H)$, Let $V_L=V\setminus V_H$. Compute a weighted fractional $\vec{f}$ matching on the low capacity edges of $G[V_L]$, such that $\vec{f}$ obeys the capacity function $\kappa$, and the odd set constraints. Then, $\sum_{e\in E}f(e)\cdot w(e)+w(M_H)\geqslant  \textsf{mwm}(G_s)-\eps\cdot n$.
\end{lemma2}

The proof of the unweighted version of this theorem relies on the extended Hall's theorem, and similar to Tutte Berge theorem, there is no weighted version of this theorem. Thus, we have to come up with a generalization that relies on a primal-dual analysis. \\ \\
Note that in the unweighted version, we can easily find such a fractional matching on low capacity edges, using approximate max-flow. However, weighted fractional matchings obeying capacities correspond to min-cost flow, and for this problem, we only have algorithms with $O(m^{1+o(1)})$ time (see \cite{BPT20,BrandCKLPGSS}) or ones where the runtime depends polynomially on the ratio between maximum and minimum capacities, which in our setting can be $O(n^2)$ (see \cite{GoldbergHKT17}). Thus, we have to come up with a simple, efficient algorithm to do this task.

\begin{lemma2}
There is an efficient algorithm to compute $(1-\eps)$-approximate maximum weight fractional matching in a capacitated bipartite graph. 
\end{lemma2}

Finally, we show how to identify the set of bottleneck edges, $E^*$.

\begin{lemma2}The set of bottleneck edges of E* can be identified by looking at the dual of the matching program of $G_s$. 
\end{lemma2}

%% file: main.bbl
\newcommand{\etalchar}[1]{$^{#1}$}
\begin{thebibliography}{vdBCK{\etalchar{+}}23}

\bibitem[ABD22]{AssadiBD22}
Sepehr Assadi, Aaron Bernstein, and Aditi Dudeja.
\newblock Decremental matching in general graphs.
\newblock In Mikolaj Bojanczyk, Emanuela Merelli, and David~P. Woodruff,
  editors, {\em 49th International Colloquium on Automata, Languages, and
  Programming, {ICALP} 2022, July 4-8, 2022, Paris, France}, volume 229 of {\em
  LIPIcs}, pages 11:1--11:19. Schloss Dagstuhl - Leibniz-Zentrum f{\"{u}}r
  Informatik, 2022.

\bibitem[AKL19]{AKL19}
Sepehr Assadi, Sanjeev Khanna, and Yang Li.
\newblock The stochastic matching problem with (very) few queries.
\newblock {\em ACM Transactions on Economics and Computation (TEAC)},
  7(3):1--19, 2019.

\bibitem[BCHN18]{BDHN18}
Sayan Bhattacharya, Deeparnab Chakrabarty, Monika Henzinger, and Danupon
  Nanongkai.
\newblock Dynamic algorithms for graph coloring.
\newblock In {\em Proceedings of the Twenty-Ninth Annual ACM-SIAM Symposium on
  Discrete Algorithms}, pages 1--20. SIAM, 2018.

\bibitem[BDL21]{BernsteinDL21}
Aaron Bernstein, Aditi Dudeja, and Zachary Langley.
\newblock A framework for dynamic matching in weighted graphs.
\newblock In Samir Khuller and Virginia~Vassilevska Williams, editors, {\em
  {STOC} '21: 53rd Annual {ACM} {SIGACT} Symposium on Theory of Computing,
  Virtual Event, Italy, June 21-25, 2021}, pages 668--681. {ACM}, 2021.

\bibitem[BGS20]{BPT20}
Aaron Bernstein, Maximilian~Probst Gutenberg, and Thatchaphol Saranurak.
\newblock Deterministic decremental reachability, scc, and shortest paths via
  directed expanders and congestion balancing.
\newblock In {\em 2020 IEEE 61st Annual Symposium on Foundations of Computer
  Science (FOCS)}, pages 1123--1134. IEEE, 2020.

\bibitem[BHI15]{BHI15}
Sayan Bhattacharya, Monika Henzinger, and Giuseppe~F. Italiano.
\newblock Deterministic fully dynamic data structures for vertex cover and
  matching.
\newblock In {\em Proceedings of the Twenty-Sixth Annual ACM-SIAM Symposium on
  Discrete Algorithms}, SODA '15, page 785–804, USA, 2015. Society for
  Industrial and Applied Mathematics.

\bibitem[BK23]{BlikstadK23}
Joakim Blikstad and Peter Kiss.
\newblock Incremental (1-{\(\epsilon\)})-approximate dynamic matching in
  o(poly(1/{\(\epsilon\)})) update time.
\newblock In Inge~Li G{\o}rtz, Martin Farach{-}Colton, Simon~J. Puglisi, and
  Grzegorz Herman, editors, {\em 31st Annual European Symposium on Algorithms,
  {ESA} 2023, September 4-6, 2023, Amsterdam, The Netherlands}, volume 274 of
  {\em LIPIcs}, pages 22:1--22:19. Schloss Dagstuhl - Leibniz-Zentrum f{\"{u}}r
  Informatik, 2023.

\bibitem[BKS23a]{BKS23}
Sayan Bhattacharya, Peter Kiss, and Thatchaphol Saranurak.
\newblock Dynamic {\textdollar}(1+{\(\epsilon\)}){\textdollar}-approximate
  matching size in truly sublinear update time.
\newblock {\em CoRR}, abs/2302.05030, 2023.

\bibitem[BKS23b]{BhattacharyaKS23b}
Sayan Bhattacharya, Peter Kiss, and Thatchaphol Saranurak.
\newblock Dynamic algorithms for packing-covering lps via multiplicative weight
  updates.
\newblock In Nikhil Bansal and Viswanath Nagarajan, editors, {\em Proceedings
  of the 2023 {ACM-SIAM} Symposium on Discrete Algorithms, {SODA} 2023,
  Florence, Italy, January 22-25, 2023}, pages 1--47. {SIAM}, 2023.

\bibitem[BS15]{BS15}
Aaron Bernstein and Cliff Stein.
\newblock Fully dynamic matching in bipartite graphs.
\newblock In Magn{\'{u}}s~M. Halld{\'{o}}rsson, Kazuo Iwama, Naoki Kobayashi,
  and Bettina Speckmann, editors, {\em Automata, Languages, and Programming -
  42nd International Colloquium, {ICALP} 2015, Kyoto, Japan, July 6-10, 2015,
  Proceedings, Part {I}}, volume 9134 of {\em Lecture Notes in Computer
  Science}, pages 167--179. Springer, 2015.

\bibitem[BS16]{BS16}
Aaron Bernstein and Cliff Stein.
\newblock Faster fully dynamic matchings with small approximation ratios.
\newblock In Robert Krauthgamer, editor, {\em Proceedings of the Twenty-Seventh
  Annual {ACM-SIAM} Symposium on Discrete Algorithms, {SODA} 2016, Arlington,
  VA, USA, January 10-12, 2016}, pages 692--711. {SIAM}, 2016.

\bibitem[CST23]{ChenST23}
Jiale Chen, Aaron Sidford, and Ta-Wei Tu.
\newblock Entropy regularization and faster decremental matching in general
  graphs.
\newblock 2023.

\bibitem[Dah16]{Dahlgaard16}
S{\o}ren Dahlgaard.
\newblock On the hardness of partially dynamic graph problems and connections
  to diameter.
\newblock {\em CoRR}, abs/1602.06705, 2016.

\bibitem[DP09]{DP09}
Devdatt~P. Dubhashi and Alessandro Panconesi.
\newblock {\em Concentration of Measure for the Analysis of Randomized
  Algorithms}.
\newblock Cambridge University Press, 2009.

\bibitem[DP14]{DP14}
Ran Duan and Seth Pettie.
\newblock Linear-time approximation for maximum weight matching.
\newblock {\em Journal of the ACM (JACM)}, 61(1):1--23, 2014.

\bibitem[GHKT17]{GoldbergHKT17}
Andrew~V. Goldberg, Sagi Hed, Haim Kaplan, and Robert~E. Tarjan.
\newblock Minimum-cost flows in unit-capacity networks.
\newblock {\em Theory Comput. Syst.}, 61(4):987--1010, 2017.

\bibitem[GLS{\etalchar{+}}19]{GLSSS19}
Fabrizio Grandoni, Stefano Leonardi, Piotr Sankowski, Chris Schwiegelshohn, and
  Shay Solomon.
\newblock {(1} + {\(\epsilon\)})-approximate incremental matching in constant
  deterministic amortized time.
\newblock In Timothy~M. Chan, editor, {\em Proceedings of the Thirtieth Annual
  {ACM-SIAM} Symposium on Discrete Algorithms, {SODA} 2019, San Diego,
  California, USA, January 6-9, 2019}, pages 1886--1898. {SIAM}, 2019.

\bibitem[GP13]{GP13}
Manoj Gupta and Richard Peng.
\newblock Fully dynamic (1+ e)-approximate matchings.
\newblock {\em 2013 IEEE 54th Annual Symposium on Foundations of Computer
  Science}, pages 548--557, 2013.

\bibitem[Gup14]{Gupta2014}
Manoj Gupta.
\newblock Maintaining approximate maximum matching in an incremental bipartite
  graph in polylogarithmic update time.
\newblock In {\em FSTTCS}, 2014.

\bibitem[HKNS15]{HKNS15}
Monika Henzinger, Sebastian Krinninger, Danupon Nanongkai, and Thatchaphol
  Saranurak.
\newblock Unifying and strengthening hardness for dynamic problems via the
  online matrix-vector multiplication conjecture.
\newblock In {\em Proceedings of the forty-seventh annual ACM symposium on
  Theory of computing}, pages 21--30, 2015.

\bibitem[JJST22]{JambulapatiJST22}
Arun Jambulapati, Yujia Jin, Aaron Sidford, and Kevin Tian.
\newblock Regularized box-simplex games and dynamic decremental bipartite
  matching.
\newblock In Mikolaj Bojanczyk, Emanuela Merelli, and David~P. Woodruff,
  editors, {\em 49th International Colloquium on Automata, Languages, and
  Programming, {ICALP} 2022, July 4-8, 2022, Paris, France}, volume 229 of {\em
  LIPIcs}, pages 77:1--77:20. Schloss Dagstuhl - Leibniz-Zentrum f{\"{u}}r
  Informatik, 2022.

\bibitem[KPP16]{KPP16}
Tsvi Kopelowitz, Seth Pettie, and Ely Porat.
\newblock Higher lower bounds from the 3sum conjecture.
\newblock In {\em Proceedings of the twenty-seventh annual ACM-SIAM symposium
  on Discrete algorithms}, pages 1272--1287. SIAM, 2016.

\bibitem[MMR94]{MMR94}
Gradimir~V. Milovanovic, Dragoslav~S. Mitrinovic, and Themistocles~M. Rassias.
\newblock {\em Topics in polynomials - extremal problems, inequalities, zeros}.
\newblock World Scientific, 1994.

\bibitem[Sch03]{Schrijver2003CombinatorialOP}
Alexander Schrijver.
\newblock Combinatorial optimization. polyhedra and efficiency.
\newblock 2003.

\bibitem[ST83]{SleatorT83}
Daniel~Dominic Sleator and Robert~Endre Tarjan.
\newblock A data structure for dynamic trees.
\newblock {\em J. Comput. Syst. Sci.}, 26(3):362--391, 1983.

\bibitem[vdBCK{\etalchar{+}}23]{BrandCKLPGSS}
Jan van~den Brand, Li~Chen, Rasmus Kyng, Yang~P. Liu, Richard Peng,
  Maximilian~Probst Gutenberg, Sushant Sachdeva, and Aaron Sidford.
\newblock A deterministic almost-linear time algorithm for minimum-cost flow.
\newblock {\em CoRR}, abs/2309.16629, 2023.

\bibitem[Waj20]{Wajc2020}
David Wajc.
\newblock Rounding dynamic matchings against an adaptive adversary.
\newblock In {\em Proceedings of the 52nd Annual ACM SIGACT Symposium on Theory
  of Computing}, pages 194--207, 2020.

\end{thebibliography}
